\documentclass[12pt,onecolumn,twoside]{IEEEtran}

\usepackage{calc}

\usepackage{multirow}
\usepackage{cite}
\usepackage{graphicx,subfigure}
\usepackage{psfrag}
\usepackage{amsmath,amssymb,amsthm}
\usepackage{color}
\usepackage{tikz} 
\usepackage{bm}
\usepackage{bbm}
\usepackage{verbatim}
\usepackage[T1]{fontenc}

\usepackage{cases}
\usepackage[noend]{algpseudocode}

\usepackage{xcolor}
\usepackage[most]{tcolorbox}
\usepackage{bm}

\usepackage{tabu}
\usepackage{mathtools}
\usepackage{xifthen}
\usepackage{booktabs}
\usepackage{algorithm}
\usepackage{algpseudocode}
\usepackage{arydshln}

\DeclarePairedDelimiter\ceil{\lceil}{\rceil}
\DeclarePairedDelimiter\floor{\lfloor}{\rfloor}

\makeatletter
\def\BState{\State\hskip-\ALG@thistlm}
\makeatother

\DeclareMathOperator*{\defeq}{\triangleq}

\newtheorem{theorem}{Theorem}

\newtheorem{lemma}{Lemma}
\newtheorem{definition}{Definition}

\newcommand{\bit}{\begin{itemize}}
\newcommand{\eit}{\end{itemize}}

\newcommand{\bc}{\begin{center}}
\newcommand{\ec}{\end{center}}

\newcommand{\ba}{\begin{array}}
\newcommand{\ea}{\end{array}}

\newcommand{\beq}{\begin{equation}}
\newcommand{\eeq}{\end{equation}}

\newcommand{\beqn}{\begin{equation*}}
\newcommand{\eeqn}{\end{equation*}}

\newcommand{\bean}{\begin{eqnarray*}}
\newcommand{\eean}{\end{eqnarray*}}
\newcommand{\bea}{\begin{eqnarray}}
\newcommand{\eea}{\end{eqnarray}}

\def\wv{\boldsymbol{w}}

\def\yv{\boldsymbol{y}}

\newcommand{\Bc}{{\mathcal B}}

\newcommand{\Fc}{{\mathcal F}}

\newcommand{\Hc}{{\mathcal H}}
\newcommand{\Ic}{{\mathcal I}}

\newcommand{\Lc}{{\mathcal L}}

\newcommand{\Rc}{{\mathcal R}}
\newcommand{\Sc}{{\mathcal S}}

\algnewcommand{\IfThenElse}[3]{  \State \algorithmicif\ #1\ \algorithmicthen\ #2\ \algorithmicelse\ #3}

\newcommand{\Me}{\wv}
\newcommand{\Ry}{\mathrm{s}}

\newcommand{\Vr}{\mathrm{v}}

\newcommand{\Ss}{\mathbb{S}}

\newcommand{\true}{\mathrm{true}}       
\newcommand{\false}{\mathrm{false}}

\newcommand{\rtuple}{)} 
\newcommand{\ltuple}{(}

\newcommand{\inset}{\in}

\newcommand{\AND}{\land}    
\newcommand{\OR}{\lor}

\newcommand{\eqlog}{=} 
\newcommand{\nodeindexstar}{i^{\star}}      
 
\newcommand{\MVBAInputMsg}{\wv}        
\newcommand{\Predicate}{\mathrm{Predicate}}

\newcommand{\EncodedSymbol}{y}

\newcommand{\vectorcommitment}{C}                    
\newcommand{\proofpositionvc}{\omega}                    
\newcommand{\VCOpen}{\mathrm{VcOpen}}                  
\newcommand{\VCCom}{\mathrm{VcCom}}  
\newcommand{\VCVerify}{\mathrm{VcVerify}}  
\newcommand{\SHARE}{\text{``}\mathrm{SHARE}\text{''}}                   
\newcommand{\ShareRecord}{\Sc_\mathrm{shares}}               
                    
\newcommand{\VOTE}{\text{``}\mathrm{VOTE}\text{''}}

\newcommand{\CONFIRM}{\text{``}\mathrm{CONFIRM}\text{''}}

\newcommand{\electionround}{r} 
   
\newcommand{\Election}{\mathrm{Election}}               
\newcommand{\electionoutput}{l}

\newcommand{\ABBA}{\mathrm{ABBA}}               
\newcommand{\ECHOSHARE}{\text{``}\mathrm{ECHOSHARE}\text{''}}                        
\newcommand{\CodedSymbols}{\mathbb{Y}_\mathrm{Symbols}}                                                                                                           
               
\newcommand{\ACD}{\mathrm{ACID}}

\newcommand{\thisnodeindex}{i}                       
\newcommand{\READY}{\text{``}\mathrm{READY}\text{''}}                           
\newcommand{\send}{\textbf{send}}                   
                    
\newcommand{\ReadyRecord}{\Rc_\mathrm{ready}}                                                                                                           
                        
\newcommand{\IDMVBA}{\mathrm{ID}}    
  
\newcommand{\IDMVBAtilde}{\mathrm{id}}                                                                  
\newcommand{\jstar}{j^{\star}}                       
                   
\newcommand{\HashRecord}{\Hc_\mathrm{hash}}                                                                                                           
\newcommand{\LockRecord}{\Lc_\mathrm{lock}}                                                                                                           
\newcommand{\FinishRecord}{\Fc_\mathrm{finish}}                                                                                                           
\newcommand{\LOCK}{\text{``}\mathrm{LOCK}\text{''}}                       
\newcommand{\ELECTION}{\text{``}\mathrm{ELECTION}\text{''}}                        
\newcommand{\FINISH}{\text{``}\mathrm{FINISH}\text{''}}                        
                        
\newcommand{\ABBAoutput}{a}   
  
\newcommand{\ABBBA}{\mathrm{ABBBA}}       
\newcommand{\ABAoutput}{b}    
                    
\newcommand{\Output}{\textbf{output}}                      
\newcommand{\terminate}{\textbf{terminate}}                        
\newcommand{\DataRetrieval}{\mathrm{DR}}        
\newcommand{\share}{\mathrm{share}}        
\newcommand{\lockindicator}{\mathrm{lock}\_\mathrm{indicator}}         
\newcommand{\abbainput}{a}     
\newcommand{\abbainputA}{a_1}    
\newcommand{\abbainputB}{a_2}      
\newcommand{\ABBAVALUE}{\text{``}\mathrm{ABBA}\text{''}}                        
                        
\newcommand{\ABBACountA}{\mathrm{cnt}_1}    
\newcommand{\ABBACountB}{\mathrm{cnt}_2}      
\newcommand{\ABBACountC}{\mathrm{cnt}_3}      
         
\newcommand{\OciorRMVBA}{\mathrm{OciorMVBA}}     
\newcommand{\RecursiveMVBA}{\mathrm{RMVBA}}        
\newcommand{\groupindex}{p}     
\newcommand{\MVBA}{\mathrm{MVBA}} 
\newcommand{\wait}{\textbf{wait}}                             
                        
\newcommand{\ECC}{\mathrm{ECC}}  
\newcommand{\EC}{\mathrm{EC}}  
\newcommand{\ECEnc}{\mathrm{ECEnc}}   
\newcommand{\ECDec}{\mathrm{ECDec}}    
\newcommand{\ECCEnc}{\mathrm{ECCEnc}}   
\newcommand{\ECCDec}{\mathrm{ECCDec}} 
\newcommand{\Alphabet}{\Bc}      
\newcommand{\CCbits}{f_\mathrm{TB}}        
\newcommand{\MVBAOutputMsg}{\hat{\wv}}           
\newcommand{\IEMVBA}{\mathrm{IneMVBA}}

\newcommand{\setindex}{\theta}   
\newcommand{\setindexnew}{j}   
\newcommand{\DRBCReadyindicator}{I_\mathrm{ready}}  
\newcommand{\DRBCFinishindicator}{I_\mathrm{finish}}   
\newcommand{\DRBCConfirmindicator}{I_\mathrm{confirm}}   
\newcommand{\COOLDRBC}{\text{DRBC-COOL}}  
\newcommand{\networksize}{\tilde{n}}   
\newcommand{\networkfaultsize}{\tilde{t}}   
\newcommand{\networksizestar}{n^{\star}}   
\newcommand{\networkfaultsizestar}{t^{\star}}    

\newcommand{\INITIAL}{\text{``}\mathrm{INITIAL}\text{''}}

\newcommand{\OECsymbolset}{\mathbb{Z}_{\mathrm{oec}}}   
\newcommand{\Lkset}{\mathbb{U}}   
\newcommand{\SYMBOL}{\text{``}\mathrm{SYMBOL}\text{''}}                        
\newcommand{\SIone}{\text{``}\mathrm{SI1}\text{''}}                         
\newcommand{\SItwo}{\text{``}\mathrm{SI2}\text{''}}

\newcommand{\defaultvalue}{\bot}      
\newcommand{\ECCEncindicator}{I_\mathrm{ecc}}

\newcommand{\Phtwoindicator}{I_\mathrm{2}}        
\newcommand{\Phthreeindicator}{I_\mathrm{3}}

\newcommand{\CORRECTSYMBOL}{\text{``}\mathrm{CORRECT}\text{''}}                         
\newcommand{\OECCorrectSymbolSet}{\mathbb{Y}_{\mathrm{oec}}}    
\newcommand{\kstar}{k^{\star}}         
\newcommand{\ktilde}{\tilde{k}}           
          
\newcommand{\OEC}{\mathrm{OEC}}

\newcommand{\Pass}{\textbf{pass}}

\newcommand{\COOL}{\mathrm{COOL}}

\newcommand{\DM}{\mathrm{DM}}

\newcommand{\networksizen}{n}

\newcommand{\successindicator}{\mathrm{s}}

\newcommand{\IDCOOL}{\mathrm{ID}}                                                                       
\newcommand{\OciorCOOL}{\text{OciorCOOL}}

\newcommand{\BA}{\mathrm{BA}}

\newcommand{\OciorRBC}{\text{OciorRBC}}

\newcommand{\ACDRRIT}{\ACD\text{rr}}                               
                  
\newcommand{\ABA}{\mathrm{ABA}}

\newcommand{\DataRetrievalRelaxedResilienceIT}{\text{DRrr}}                   
    
\newcommand{\AMBARelaxedResilienceIT}{\ABA\text{rr}}                  
\newcommand{\Byzantineuniqueagreement}{\text{Unique  agreement}}            
\newcommand{\BUA}{\mathrm{UA}}

\newcommand{\RBA}{\mathrm{RBA}}

\newcommand{\HMDM}{\mathrm{HMDM}}

\newcommand{\RBC}{\mathrm{RBC}}

\newcommand{\OECsymbol}{z}                                                                            
\newcommand{\OECsymbolsetInitial}{\mathbb{Z}_{\mathrm{oec}}}     
                                                                        
\newcommand{\SIPhtwo}{I_\mathrm{SI2}}     
\newcommand{\OECSIFinal}{I_\mathrm{oecfinal}}                                                                              
         
\newcommand{\VrOutput}{\mathrm{v}_o} 
\newcommand{\OciorRBA}{\text{OciorRBA}}            
\newcommand{\RMVBA}{\mathrm{RMVBA}}              
\newcommand{\deliver}{\textbf{deliver}}                     
\newcommand{\SHMDM}{\mathrm{SHMDM}}

\newcommand{\alphabetsize}{q}            
                                           
\newcommand{\byzantineuniqueagreement}{\text{unique  agreement}}

\newcommand{\Node}{P}

 \newcommand{\networkfaultsizereal}{t}     
 \newcommand{\Round}{r}         
 \newcommand{\OciorMVBAHash}{\mathrm{OciorMVBAh}}

\newcommand{\OciorMVBARelaxedResilienceIT}{\mathrm{OciorMVBArr}}    
\newcommand{\ACDh}{\mathrm{ACIDh}}                    
\newcommand{\DataRetrievalh}{\mathrm{DRh}}

\algdef{SE}[SUBALG]{Indent}{EndIndent}{}{\algorithmicend\ }%
\algtext*{Indent}
\algtext*{EndIndent}

\begin{document}
\sloppy
\title{OciorMVBA: Near-Optimal Error-Free   Asynchronous MVBA}

\author{Jinyuan Chen  
}

\maketitle
\pagestyle{headings}

\begin{abstract}

In this work, we propose an error-free, information-theoretically secure, asynchronous multi-valued validated Byzantine agreement ($\MVBA$) protocol, called $\OciorRMVBA$. This protocol achieves $\MVBA$ consensus on a message $\MVBAInputMsg$ with expected $O(n |\MVBAInputMsg|\log n + n^2 \log \alphabetsize)$ communication bits,  expected $O(n^2)$ messages,   expected  $O(\log n)$ rounds, and   expected $O(\log n)$  common coins, under optimal resilience $n \geq 3t + 1$ in an $n$-node network, where up to $t$ nodes may be dishonest. Here, $\alphabetsize$ denotes the alphabet size of the error correction code used in the protocol. When error correction codes with a constant alphabet size (e.g., Expander Codes) are used, $\alphabetsize$ becomes a constant. An $\MVBA$ protocol that guarantees all required properties without relying on any cryptographic assumptions, such as signatures or hashing, except for the common coin assumption, is said to be \emph{information-theoretically secure (IT secure)}. Under the common coin assumption, an $\MVBA$ protocol that guarantees all required properties in \emph{all} executions is said to be \emph{error-free}.

We also propose another error-free, IT-secure, asynchronous $\MVBA$ protocol, called $\OciorMVBARelaxedResilienceIT$. This protocol achieves $\MVBA$ consensus with  expected $O(n |\MVBAInputMsg| + n^2 \log n)$ communication bits,   expected   $O(1)$ rounds, and   expected $O(1)$ common coins, under a relaxed resilience (RR) of $n \geq 5t + 1$. 
Additionally, we propose a hash-based asynchronous $\MVBA$ protocol, called $\OciorMVBAHash$. This protocol achieves $\MVBA$ consensus with  expected $O(n |\MVBAInputMsg| + n^3)$ bits,   expected   $O(1)$ rounds, and   expected $O(1)$ common coins, under optimal resilience $n \geq 3t + 1$.

\end{abstract}


\section{Introduction}

 Multi-valued validated  Byzantine agreement ($\MVBA$),   introduced by Cachin et al. in 2001 \cite{CKPS:01}, is one of the key building blocks  for  distributed systems and cryptography. 
 In  $\MVBA$,  distributed nodes proposes their input values and seek to agree on one of the proposed values, ensuring that the agreed value satisfies a predefined  $\Predicate$ function  (referred to as \emph{External Validity}).   
 $\MVBA$ is a variant of Byzantine agreement ($\BA$), which was  proposed by Pease, Shostak and Lamport in 1980 \cite{PSL:80}. 
In  $\BA$, if all honest nodes input the same  value $\wv$, it is required that every honest node eventually outputs $\wv$ (referred to as \emph{Validity}).
 As one can see,  $\MVBA$'s  External Validity is different from $\BA$'s Validity.

In this work, we  focus on the design of  \emph{asynchronous} $\MVBA$ protocols. The seminal work by Fischer, Lynch, and Paterson   \cite{FLP:85} reveals that no deterministic $\MVBA$ protocol can exist in the asynchronous setting. Therefore, any asynchronous $\MVBA$ protocol must incorporate randomness. A common approach to designing such a protocol is to create a deterministic algorithm supplemented by common coins, which provide the necessary randomness.
  
Additionally, we primarily    focus on   the design of error-free, information-theoretically secure (IT secure), asynchronous $\MVBA$ protocols. 
 An $\MVBA$ protocol that guarantees all required properties without relying on any cryptographic assumptions, such as signatures or hashing, except for the common coin assumption, is said to be \emph{IT secure}. Under the common coin assumption, an $\MVBA$ protocol that guarantees all required properties in \emph{all} executions is said to be \emph{error-free}.

 Specifically, we propose an error-free, IT secure, asynchronous $\MVBA$ protocol, called $\OciorRMVBA$. This protocol achieves $\MVBA$ consensus on a message $\MVBAInputMsg$ with expected $O(n |\MVBAInputMsg|\log n + n^2 \log \alphabetsize)$ communication bits,  expected $O(n^2)$ messages,   expected  $O(\log n)$ rounds, and   expected $O(\log n)$  common coins, under optimal resilience $n \geq 3t + 1$ in an $n$-node network, where up to $t$ nodes may be dishonest. Here, $\alphabetsize$ denotes the alphabet size of the error correction code used in the protocol. When error correction codes with a constant alphabet size (e.g., Expander Codes \cite{SS:96}) are used, $\alphabetsize$ becomes a constant.  
 
 The design of $\OciorRMVBA$  in this $\MVBA$ setting builds on the protocols   $\COOL$ and $\OciorCOOL$, originally designed for the $\BA$ setting  \cite{Chen:2020arxiv, ChenDISC:21,ChenOciorCOOL:24}.
 Specifically,    $\COOL$ and $\OciorCOOL$ introduced two primitives:   $\byzantineuniqueagreement$ ($\BUA$)   and honest-majority distributed multicast ($\HMDM$) \cite{Chen:2020arxiv, ChenDISC:21,ChenOciorCOOL:24}.   
 $\COOL$ and $\OciorCOOL$ achieve the deterministic, error-free, IT secure, synchronous $\BA$ consensus with   $O(n |\MVBAInputMsg|  + n t \log \alphabetsize)$ communication bits and    $O(t)$ rounds,  under optimal resilience $n \geq 3t + 1$.  When error correction codes with a constant alphabet size (e.g., Expander Code \cite{SS:96}) are used, $\alphabetsize$ becomes a constant, and consequently,  $\COOL$ and $\OciorCOOL$  are optimal.

$\bullet$ $\BUA$:    In  $\BUA$, distributed nodes input their values and  seek to decide on an  output of the form  $(\MVBAInputMsg, \successindicator, \Vr)$, where $\successindicator \in \{0,1\}$ and $\Vr \in \{0,1\}$ denote  a success indicator and  a vote, respectively.    $\BUA$  requires  that all honest nodes  eventually output the same value  $\MVBAInputMsg$ or a default value (\emph{Unique Agreement}).  Furthermore,   if any honest node votes $\Vr=1$, then at least $t+1$ honest nodes eventually output  $(\wv, 1, *)$ for the same $\wv$ (\emph{Majority Unique Agreement}).

 $\bullet$ $\HMDM$:  In   $\HMDM$, there are at least $t+1$ honest nodes acting as senders, multicasting a message to $n$ nodes. 
  $\HMDM$ requires  that   if every  honest sender inputs the same message $\wv$,   then every honest node eventually outputs $\wv$.       
 
  Our proposed  $\OciorRMVBA$ is a recursive protocol (see Fig.~\ref{fig:OciorRMVBA}) that consists of  the algorithms of strongly-honest-majority distributed multicast ($\SHMDM$), reliable Byzantine agreement ($\RBA$), asynchronous  biased  binary Byzantine agreement ($\ABBBA$),  and  asynchronous  binary  BA   ($\ABBA$).

$\bullet$ $\RBA$:    Our $\RBA$ algorithm, called  $\OciorRBA$,  is built from $\BUA$ and $\HMDM$.   It is worth noting that the   Unique Agreement and    Majority Unique Agreement  properties of  $\BUA$  provide an ideal condition for $\HMDM$.  
 
  $\bullet$ $\SHMDM$:      $\SHMDM$ is slightly different from $\HMDM$, as    all honest nodes  act as senders in $\SHMDM$.  
  
  $\bullet$ $\ABBBA$:     This new primitive is introduced here and used as a building block in our protocols. In $\ABBBA$,  each  honest node inputs a pair of binary numbers $(\abbainputA, \abbainputB)$, for some $\abbainputA, \abbainputB \in \{0,1\}$. One property is that if $t+1$ honest nodes input the second number as $\abbainputB=1$, then any honest node that terminates outputs $1$ (Biased Validity). Another property is that  
  if an honest node outputs $1$, then at least one honest node inputs $\abbainputA=1$ or $\abbainputB=1$  (Biased Integrity).

 In this work, we also propose another error-free, IT-secure, asynchronous $\MVBA$ protocol, called $\OciorMVBARelaxedResilienceIT$. This protocol achieves $\MVBA$ consensus with  expected $O(n |\MVBAInputMsg| + n^2 \log n)$ communication bits,   expected   $O(1)$ rounds, and   expected $O(1)$ common coins, under a relaxed resilience (RR) of $n \geq 5t + 1$. 
Additionally, we propose a hash-based asynchronous $\MVBA$ protocol, called $\OciorMVBAHash$. This protocol achieves $\MVBA$ consensus with  expected $O(n |\MVBAInputMsg| + n^3)$ bits,   expected   $O(1)$ rounds, and   expected $O(1)$ common coins, under optimal resilience $n \geq 3t + 1$.

The proposed $\OciorRMVBA$ protocol is described in Algorithms~\ref{algm:OciorRMVBA}-\ref{algm:OciorRBA} and Section~\ref{sec:OciorRMVBA}. 
The proposed $\OciorMVBARelaxedResilienceIT$ protocol is described in Algorithms~\ref{algm:OciorMVBARelaxedResilienceIT}-\ref{algm:DataRetrievalRelaxedResilienceIT}  and  Section~\ref{sec:OciorMVBARelaxedResilienceIT}.  
The proposed $\OciorMVBAHash$ protocol is described in Algorithms~\ref{algm:OciorMVBAh}-\ref{algm:DataRetrievalh}  and  Section~\ref{sec:OciorMVBAHash}. 
 Table~\ref{tb:MVBA}   provides a comparison between the proposed protocols and some other $\MVBA$ protocols. 
In the following subsection, we provide some definitions and primitives  used in our protocols.

{\renewcommand{\arraystretch}{1.3}
\begin{table}
\footnotesize  
\begin{center}
\caption{Comparison between the proposed protocols and some other $\MVBA$ protocols.  Here $\alphabetsize$ denotes the alphabet size of  the error correction code used in the proposed  protocols. 
When error correction codes with a constant alphabet size (e.g., Expander Code \cite{SS:96}) are used, $\alphabetsize$ becomes a constant.   $\kappa$ is a security parameter. 
} \label{tb:MVBA}

\begin{tabular}{||c||c|c|c|c|c|}
\hline
Protocols & Resilience &   Communication     &  $\#$ Coin  &       Rounds     &   Cryptographic     Assumption    \\ 
  &  &  (Total Bits)  &     &               &        (Expect for     Common Coin)       \\ 
\hline
Cachin et al. \cite{CKPS:01}  &  $t<\frac{n}{3}$  &    $O(n^2|\MVBAInputMsg|+ \kappa n^2+n^3)$    &  $O(1)$    & $O(1)$   &  Threshold Sig   \\
\hline
Abraham et al.    \cite{AMS:19} &  $t<\frac{n}{3}$   &   $O(n^2|\MVBAInputMsg|+ \kappa n^2)$      &  $O(1)$     &  $O(1)$  &  Threshold Sig     \\
\hline
Dumbo-MVBA \cite{LLTW:20} &  $t<\frac{n}{3}$   &    $O(n|\MVBAInputMsg|+ \kappa n^2)$       &   $O(1)$     &  $O(1)$   &  Threshold Sig        \\ 
\hline
\hline
Duan et al.\cite{DWZ:23} &  $t<\frac{n}{3}$ &    $O(n^2|\MVBAInputMsg|+ \kappa n^3)$        & $O(1)$     &  $O(1)$  &  Hash       \\ 
\hline
Feng et al.\cite{FLMT:24} &  $t<\frac{n}{5}$ &    $O(n|\MVBAInputMsg|+ \kappa n^2\log n)$       & $O(1)$    &  $O(1)$  &  Hash      \\ 
\hline
Komatovic et al.  \cite{KNR:24}  &  $t<\frac{n}{4}$ &    $O(n|\MVBAInputMsg|+ \kappa n^2\log n)$        & $O(1)$    &  $O(1)$  &  Hash      \\ 
\hline
 {\color{blue} Proposed $\OciorMVBAHash$} &   {\color{blue} $t<\frac{n}{3}$}  &   {\color{blue} $O(n|\MVBAInputMsg|+ \kappa n^3)$  }     &   {\color{blue} $O(1)$}    &   {\color{blue} $O(1)$}    &    {\color{blue} Hash}          \\
\hline
\hline
Duan et al.\cite{DWZ:23} &  {\color{blue}  $t<\frac{n}{3}$}  &    $O(n^2|\MVBAInputMsg|+  n^3\log n)$   &   $O(1)$    &  $O(1)$   &  Non      \\ 
\hline
 {\color{blue}Proposed  $\OciorMVBARelaxedResilienceIT$ } &   {\color{blue} $t<\frac{n}{5}$}  &   {\color{blue} $O(n |\MVBAInputMsg|+  n^2\log n)$  }  &  {\color{blue} $O(1)$}     &   {\color{blue} $O(1)$}     &    {\color{blue} Non}           \\
\hline
 {\color{blue} Proposed $\OciorRMVBA$ } &   {\color{blue} $t<\frac{n}{3}$}  &   {\color{blue} $O(n |\MVBAInputMsg|\log n+  n^2\log \alphabetsize   )$  }   &   {\color{blue} $O(\log n)$}     &   {\color{blue} $O(\log n)$}     &    {\color{blue} Non}            \\
\hline
\end{tabular}
\end{center}
\end{table}
}

\subsection{Primitives}

 \noindent  {\bf Asynchronous network.}  We consider a network of $n$ distributed nodes, where  up to $t$ of the  nodes may be  dishonest.   Every pair of nodes is connected via a reliable and private communication channel.  
The network is considered to be   \emph{asynchronous}, i.e., the adversary can arbitrarily delay any message, but the messages sent between honest nodes will eventually arrive at their destinations.

 \noindent  {\bf Adaptive adversary.}   We consider an adaptive adversary, i.e., the adversary can corrupt any node at any time during the course of protocol execution, but at most $t$ nodes in total can be controlled by the adversary.

    \noindent  {\bf Information-theoretic protocol.} A protocol that guarantees all  required properties without relying on any   cryptographic assumptions, such as  signatures or hashing,  except for  the common coin assumption, is said to be \emph{information-theoretically secure}.    The proposed protocols $\OciorRMVBA$ and $\OciorMVBARelaxedResilienceIT$    are IT secure.

\noindent  {\bf Signature-free protocol.}   Under the common coin assumption, a protocol that guarantees all  required properties without relying on signature-based   cryptographic assumptions is said to be \emph{signature-free}.    All of the the proposed protocols   are signature-free.

\noindent  {\bf Error-free protocol.}    Under the common coin assumption, a protocol that  that guarantees all of the required properties in \emph{all} executions is said to be \emph{error-free}.    The proposed protocols $\OciorRMVBA$ and $\OciorMVBARelaxedResilienceIT$  are error-free.

\begin{definition} [{\bf Multi-valued validated Byzantine agreement ($\MVBA$)}]     \label{def:OciorMVBA}  
In the $\MVBA$ problem,  there is an external $\Predicate$ function $\{0,1\}^{*}\to \{\true, \false\}$ known to all nodes. In this problem, each honest node proposes its input value, ensuring that it satisfies the $\Predicate$ function to be true.    
The $\MVBA$ protocol guarantees  the following properties: 
 \begin{itemize}
\item  {\bf Agreement:} If any two honest nodes output $\wv'$ and $\wv''$, respectively, then  $\wv'=\wv''$.  
\item  {\bf Termination:} Every honest node eventually outputs a value and terminates.     
\item  {\bf External validity:} If an honest node outputs a value $\MVBAInputMsg$, then $\Predicate(\MVBAInputMsg)=\true$.        
\end{itemize} 
\end{definition}

\begin{definition}  [{\bf Byzantine agreement ($\BA$)}]
In the $\BA$  protocol, the distributed nodes   seek to reach agreement on a common value. 
The  $\BA$ protocol guarantees  the following properties: 
\begin{itemize}
\item  {\bf Termination:} If all  honest nodes receive their inputs, then every honest node  eventually outputs a value and terminates. 
\item  {\bf Consistency:} If any honest node output a value $\wv$, then every honest node eventually outputs $\wv$.
\item  {\bf Validity:}     If all honest nodes input the same  value $\wv$, then every honest node eventually outputs $\wv$.  
\end{itemize} 
\end{definition}

\begin{definition} [{\bf Reliable broadcast ($\RBC$)}]
 In a reliable broadcast protocol, a leader  inputs a value  and broadcasts it to distributed nodes,   satisfying the following conditions:
\begin{itemize}
\item  {\bf Consistency:} If any two honest nodes output $\wv'$ and $\wv''$, respectively, then  $\wv'=\wv''$.
\item   {\bf Validity:} If the leader is  honest and inputs a value $\wv$, then every honest node eventually outputs $\wv$. 
\item  {\bf Totality:}  If one  honest node outputs a value, then every honest node  eventually outputs a value.          
\end{itemize} 
\end{definition}

\begin{definition} [{\bf Reliable Byzantine agreement ($\RBA$)}]
$\RBA$ is a  variant of $\RBC$  problem and is a  relaxed version of $\BA$ problem. 
The  $\RBA$ protocol guarantees  the following properties: 
\item  {\bf Consistency:} If any two honest nodes output $\wv'$ and $\wv''$, respectively, then  $\wv'=\wv''$.
\item   {\bf Validity:} If all honest node input the same  value $\wv$, then every honest node eventually outputs $\wv$.   
\item  {\bf Totality:}  If one  honest node outputs a value, then every honest node  eventually outputs a value.  
\end{definition}

\begin{definition} [{\bf Distributed multicast}] \label{def:DM}
 In the problem of distributed multicast ($\DM$), there exits a subset of nodes acting as  senders multicasting the message over $n$ nodes, where up to $t$ nodes could be dishonest. Each node acting as an sender has an input   message.  A protocol is called as a $\DM$ protocol if  the following property is guaranteed:  
\begin{itemize}
\item   {\bf Validity:} If all  honest senders input the same message $\wv$,   every honest node eventually outputs $\wv$.       
\end{itemize} 
{\bf Honest-majority distributed multicast} ($\HMDM$, \cite{Chen:2020arxiv, ChenDISC:21,ChenOciorCOOL:24}): A $\DM$ problem is called as honest-majority $\DM$ if  at least   $t+1$ senders are honest.  $\HMDM$ was used previously as a building block for $\COOL$ and $\OciorCOOL$ protocols   \cite{Chen:2020arxiv, ChenDISC:21,ChenOciorCOOL:24}. \\
{\bf Strongly-honest-majority distributed multicast ($\SHMDM$):} A $\DM$ problem is called as strongly-honest-majority $\DM$ if  all honest nodes are acting as senders.  
\end{definition}

\begin{definition} [{\bf  $\Byzantineuniqueagreement$} ($\BUA$, \cite{Chen:2020arxiv, ChenDISC:21,ChenOciorCOOL:24})]  \label{def:BUA}     
$\BUA$   is a variant of  Byzantine agreement  problem operated over $n$ nodes, where up to $t$ nodes may be dishonest.  
In a $\BUA$ protocol, each node inputs an initial value and  seeks to make an  output taking the form as $(\MVBAInputMsg, \successindicator, \Vr)$, where $\successindicator \in \{0,1\}$ is a    success indicator and $\Vr \in \{0,1\}$ is a vote. 
The  $\BUA$ protocol guarantees  the following properties: 
\begin{itemize}
\item  {\bf Unique Agreement:} If any two honest nodes output $(\wv', 1, *)$ and $(\wv'', 1, *)$, respectively, then  $\wv'=\wv''$.
\item  {\bf Majority Unique Agreement:} If any honest node outputs $(*, *, 1)$, then at least $t+1$ honest nodes eventually output  $(\wv, 1, *)$ for the same $\wv$. 
\item   {\bf Validity:} If all honest nodes input the same  value $\wv$,  then  all honest nodes eventually  output $(\wv, 1, 1)$. 
\end{itemize} 
 $\BUA$ was used previously as a building block for $\COOL$ and $\OciorCOOL$ protocols   \cite{Chen:2020arxiv, ChenDISC:21,ChenOciorCOOL:24}. 
\end{definition}

\noindent  {\bf Asynchronous  complete  information dispersal ($\ACD$).}   We introduce a new primitive $\ACD$. The goal of an $\ACD$ protocol is to disperse information over distributed nodes. Once a leader completes the dispersal of its proposed message, it is guaranteed that each honest node could  retrieve the delivered message correctly from distributed nodes via a data retrieval scheme.  
Two $\ACD$ definitions  are provided below: one for an $\ACD$ instance dispersing a message proposed by a leader, and the other one for a whole $\ACD$ protocol of running  $n$ parallel  $\ACD$ instances.

\begin{definition} [$\ACD$ instance]
In an $\ACD[\ltuple \IDMVBA, i \rtuple]$ protocol with an identity $\ltuple \IDMVBA, i \rtuple$, a  message  is proposed by $\Node_i$ (i.e., the  leader in this case) and  is dispersed over $n$ distributed nodes, for $i\in [1:n]$. 
An $\ACD[\ltuple \IDMVBA, i \rtuple]$ protocol is complemented by a data retrieval protocol $\DataRetrieval[\ltuple \IDMVBA, i \rtuple]$  in which each node retrieves the   message proposed by $\Node_i$ from $n$ distributed nodes. 
The $\ACD[\ltuple \IDMVBA, i \rtuple]$ and $\DataRetrieval[\ltuple \IDMVBA, i \rtuple]$ protocols guarantee  the following properties:    
\begin{itemize}
\item  {\bf Completeness:}  If $\Node_i$ is honest, then $\Node_i$ eventually completes the dispersal $\ltuple \IDMVBA, i \rtuple$. 
\item  {\bf Availability:} If $\Node_i$ completes the dispersal for $\ltuple \IDMVBA, i \rtuple$,  and all honest nodes start the data retrieval  protocol  for $\ltuple \IDMVBA, i \rtuple$, then each node eventually reconstructs some message.
\item   {\bf Consistency:} If two honest nodes reconstruct messages $\MVBAInputMsg'$ and $\MVBAInputMsg''$ respectively for $\ltuple \IDMVBA, i \rtuple$, then  $\MVBAInputMsg'=\MVBAInputMsg''$. 
\item   {\bf Validity:} If an honest $\Node_i$ has proposed a message $\MVBAInputMsg$ for $\ltuple \IDMVBA, i \rtuple$ and an honest node reconstructs a message $\MVBAInputMsg'$  for $\ltuple \IDMVBA, i \rtuple$, then  $\MVBAInputMsg'=\MVBAInputMsg$. 
\end{itemize} 
\end{definition}

\begin{definition} [Parallel $\ACD$ instances]
An $\ACD[ \IDMVBA  ]$ protocol is a protocol involves running  $n$ parallel  $\ACD$ instances, $\{\ACD[\ltuple \IDMVBA, i \rtuple]\}_{i=1}^{n}$, over $n$ distributed nodes, where up to $t$ of the nodes  may be dishonest. 
For an $\ACD[ \IDMVBA  ]$ protocol, the following conditions must  be satisfied: 
\begin{itemize}
\item  {\bf Termination:} Every honest node eventually terminates.     
\item  {\bf Integrity:} If one honest node terminates, then there exists a set $\Ic^{\star}$ such that the following conditions hold: 1) $\Ic^{\star}\subseteq [1:n]\setminus \Fc$, where $\Fc$ denotes the set of indexes of all dishonest nodes; 2)  $|\Ic^{\star}| \geq n-2t$; and  3) for any $i\in \Ic^{\star}$,  $\Node_i$ has completed the dispersal $\ACD[\ltuple \IDMVBA, i \rtuple]$.      
\end{itemize}  
\end{definition}

\begin{definition} [\bf Asynchronous  biased  binary Byzantine agreement ($\ABBBA$)]    
We introduce a new primitive called as $\ABBBA$.   In an $\ABBBA$ protocol, each  honest node inputs a pair of binary numbers $(\abbainputA, \abbainputB)$, for some $\abbainputA, \abbainputB \in \{0,1\}$. The  honest nodes seek to reach an agreement on a common value $\abbainput \in \{0,1\}$.  
An $\ABBBA$ protocol should satisfy the following properties:
\begin{itemize}
\item   {\bf Conditional termination:} Under an input condition---i.e.,  if one honest node inputs its second number as $\abbainputB =1$ then at least $t+1$ honest nodes  input  their first numbers as $\abbainputA =1$---then every honest node eventually outputs a value and terminates.     
\item   {\bf Biased validity:} If $t+1$ honest nodes input the second number as $\abbainputB=1$, then any honest node that terminates outputs $1$.       
\item   {\bf Biased integrity:} If an honest node outputs $1$, then at least one honest node inputs $\abbainputA=1$ or $\abbainputB=1$.             
\end{itemize} 
\end{definition}

 \begin{definition} [{\bf Common coin}]
The seminal work by Fischer, Lynch, and Paterson in \cite{FLP:85} reveals that no deterministic $\MVBA$ protocol can exist in the asynchronous setting. Therefore, any asynchronous $\MVBA$ protocol must incorporate randomness. A common approach to designing such a protocol is to create a deterministic algorithm supplemented by common coins, which provide the necessary randomness.  
Here, we assume the existence of a common coin protocol $\electionoutput \gets \Election[\IDMVBAtilde]$ associated with an identity $\IDMVBAtilde$, which guarantees the following properties: 
 \begin{itemize}
\item   {\bf Termination:} If $t+1$ honest nodes activate $\Election[\IDMVBAtilde]$, then each honest node that activates it will output a common value $\electionoutput$.     
\item  {\bf Consistency:} If any two honest nodes output $\electionoutput'$ and $\electionoutput''$ from $\Election[\IDMVBAtilde]$, respectively, then  $\electionoutput'=\electionoutput''$.
\item   {\bf Uniform:} The output $\electionoutput$ from $\Election[\IDMVBAtilde]$ is randomly generated based on a uniform distribution for $\electionoutput \in [1:n]$.
\item   {\bf Unpredictability:} The adversary cannot correctly predict the output of $\Election[\IDMVBAtilde]$ unless at least one honest node has activated it.
\end{itemize} 
When analyzing the performance of $\MVBA$ protocols, we exclude the cost of the common coin protocol.
\end{definition}

 \noindent  {\bf Error correction code ($\ECC$).}   An $(n, k)$ error correction coding  scheme  consists of an encoding scheme $\ECCEnc: \Alphabet^{k} \to  \Alphabet^{n}$ and a decoding scheme $\ECCDec: \Alphabet^{n'} \to  \Alphabet^{k}$, where $\Alphabet$ denotes the alphabet of each symbol and $\alphabetsize\defeq|\Alphabet|$  denotes the size of $\Alphabet$, for some $n'$. 
 While $[\EncodedSymbol_1,  \EncodedSymbol_2, \cdots, \EncodedSymbol_{n}] \gets \ECCEnc (n,  k, \MVBAInputMsg)$ outputs $n$ encoded symbols, 
$\EncodedSymbol_{j} \gets \ECCEnc_{j}(n,  k, \MVBAInputMsg)$ outputs the $j$th encoded symbol.  

 Reed-Solomon (RS) codes (cf.~\cite{RS:60}) are widely used   error correction codes.  An $(n, k)$ RS  error correction code can correct up to  $t$ Byzantine errors and simultaneously detect up to $e$ Byzantine errors in $n'$ symbol observations, given the conditions of  $2t+ e +k \leq   n'$  and  $n' \leq  n$.  
 The $(n, k)$ RS code is operated over Galois Field $GF(\alphabetsize)$ under the constraint  $n \leq \alphabetsize  -1$ (cf.~\cite{RS:60}). 
RS codes can be constructed using \emph{Lagrange polynomial interpolation}.  The resulting code is a type of RS code with a minimum distance   $d=n-k+1$, which is optimal according to the Singleton bound. 
Berlekamp-Welch algorithm and Euclid's algorithm are two efficient decoding algorithms for RS codes \cite{roth:06, Berlekamp:68, RS:60}. 

Although RS is a popular error correction code, it has a constraint on the size of the alphabet, namely $n \leq \alphabetsize - 1$. To overcome this limitation, other error correction codes with a constant alphabet size, such as Expander Codes \cite{SS:96}, can be used.

\noindent   {\bf Erasure code  ($\EC$).}    An $(n, k)$ erasure coding  scheme  consists of an encoding scheme $\ECEnc: \Alphabet^{k} \to  \Alphabet^{n}$ and a decoding scheme $\ECDec: \Alphabet^{k} \to  \Alphabet^{k}$, where $\Alphabet$ denotes the alphabet of each symbol and $\alphabetsize\defeq|\Alphabet|$  denotes the size of $\Alphabet$. With an $(n, k)$ erasure code, the original message can be decoded from any $k$ encoded symbols. Specifically, given  $[\EncodedSymbol_1,  \EncodedSymbol_2, \cdots, \EncodedSymbol_{n}] \gets \ECEnc(n, k, \MVBAInputMsg)$, then $\ECDec(n,k, \{\EncodedSymbol_{j_1}, \EncodedSymbol_{j_2}, \cdots  \EncodedSymbol_{j_k}\}) =\MVBAInputMsg$ holds true for any $k$ distinct integers $j_1, j_2, \cdots, j_k \in [1:n]$.

 \noindent  {\bf Online error  correction  ($\OEC$).}     
Online error  correction is a variant of traditional error correction \cite{BCG:93}.   An $(n, k)$ error correction code can correct up to  $t'$ Byzantine errors in $n'$ symbol observations, provided the conditions of  $2t'+ k \leq   n'$  and  $n' \leq  n$. However, in an asynchronous setting, a node might not be able to decode the message with $n'$ symbol observations if  $2t'+ k >   n'$. 
In such a case, the node can wait for one more symbol observation before attempting to decode again. This process repeats until the node successfully decodes the message. By setting the threshold as $n' \geq k + t$, $\OEC$ may perform up to $t$ trials in the worst case before decoding the message.

\begin{figure}[H]
\centering
\includegraphics[width=18cm]{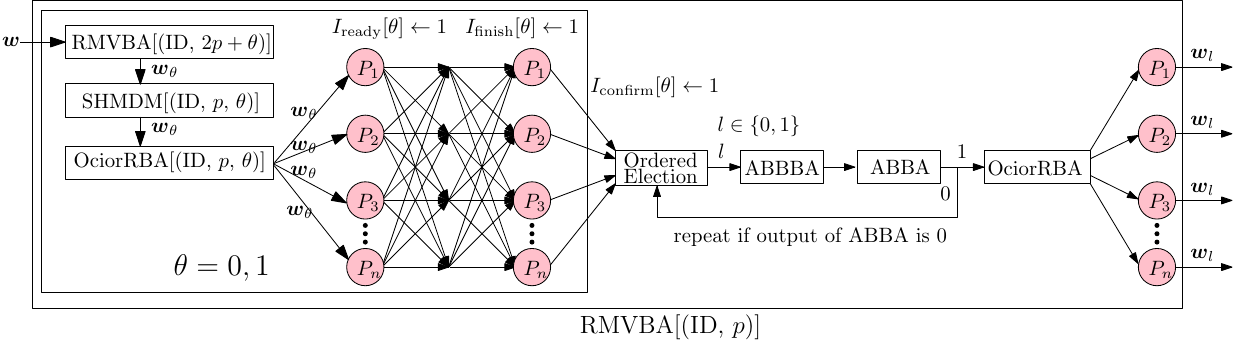}
\caption{A block diagram of the proposed $\OciorRMVBA$ protocol with an identifier $\IDMVBA$. 
}
\label{fig:OciorRMVBA}
\end{figure}

\begin{algorithm}[H]
\caption{$\OciorRMVBA$  protocol, with identifier $\IDMVBA$, for $n\geq 3t+1$. Code is shown for $\Node_{\thisnodeindex}$.   }    \label{algm:OciorRMVBA} 
\begin{algorithmic}[1]
\vspace{5pt}    
\footnotesize

\Statex   \emph{//   **   $\RMVBA$:  Recursive $\MVBA$, error-free, IT secure **}

\Statex   \emph{//   **  $\Sc_\groupindex$ is partitioned into two disjoint sets $\Sc_{2\groupindex}$ and  $\Sc_{2\groupindex+1}$ such that  $|\Sc_{2\groupindex}| =\floor{|\Sc_{\groupindex}|/2}$ and  $|\Sc_{2\groupindex+1}| =\ceil{|\Sc_{\groupindex}|/2}$**}

\Statex   \emph{//   **  Initially $\Sc_1 $  denotes the set of all nodes with size $n$ **}

\Statex
 
\Procedure{$\RecursiveMVBA[(\IDMVBA, \groupindex)]$}{$\MVBAInputMsg$}   
	\State let $\networksize \gets  |\Sc_{\groupindex}|;  \networkfaultsize\gets \floor{\frac{|\Sc_{\groupindex}|-1}{3}};    \Round_\groupindex \gets \floor{\log \groupindex} +1$  
\State let  $\setindex \gets 0, \bar{\setindex} \gets 1$ if  $\Node_{\thisnodeindex} \in \Sc_{2\groupindex}$, else $\setindex \gets 1, \bar{\setindex} \gets 0$      \label{line:OciorRMVBAsetindex} 	
 \State let $\DRBCFinishindicator \gets \{\}; \DRBCReadyindicator\gets\{\}; \DRBCConfirmindicator\gets\{\}$
	\For {$j \inset \{0,1\}$}	 
		\State $\DRBCFinishindicator[j]\gets 0;\DRBCReadyindicator[j]\gets 0; \DRBCConfirmindicator[j]\gets 0$  	
	\EndFor

\State {\bf upon} receiving  input  $\MVBAInputMsg$, for $\Predicate(\MVBAInputMsg) \eqlog \true$ and $\Node_{\thisnodeindex} \in \Sc_\groupindex$ {\bf do}:  	
\Indent	 
	\If {$|\Sc_\groupindex| \leq M$}		      \quad \quad \quad \quad \quad \quad   \quad \quad \quad\quad \quad\quad \quad\quad\quad \ \    \emph{//  $M$ is a preset finite number }  
		\State  $\MVBAOutputMsg \gets \IEMVBA[(\IDMVBA, \groupindex)](\MVBAInputMsg)$       \quad \quad  \quad \quad \quad\quad\quad\quad\quad \  \emph{//  $\IEMVBA$: inefficient $\MVBA$ protocol}      
		\State $\Return$   $\MVBAOutputMsg$ and $\terminate$ 	
	\Else 
		\State  $\Pass$ $\MVBAInputMsg$ into $\RecursiveMVBA[(\IDMVBA, 2\groupindex+\setindex)] $ as input          \label{line:OciorMVBAstarRecursivecall} 	 \quad\quad  \quad  \emph{// From Line~\ref{line:OciorRMVBAsetindex},   it is true that $\Node_{\thisnodeindex} \in \Sc_{2\groupindex+\setindex}$ and $\Node_{\thisnodeindex} \notin \Sc_{2\groupindex+\bar{\setindex}}$}          	
		\State  $\Pass$ $\bot$ into $\SHMDM[(\IDMVBA, \groupindex, \bar{\setindex})]$ as input       
		\State $\wait$ for $(\DRBCConfirmindicator[0]=1)\OR (\DRBCConfirmindicator[1]=1)$  	     				
		\For {$\electionoutput \inset \{0, 1\}$}	   \label{line:OciorRMVBALoopCond} 	
			\State  $\ABBAoutput\gets \ABBBA[\ltuple \IDMVBA, \groupindex,  \electionoutput, \networksize,  \networkfaultsize\rtuple](\DRBCReadyindicator[\electionoutput], \DRBCFinishindicator[\electionoutput])$       \label{line:OciorRMVBAABBBAoutput} 		    \     \emph{//  asynchronous  biased binary BA,  within $\Sc_{\groupindex}$, with $\networksize, \networkfaultsize$ parameters }   
			\State  $\ABAoutput\gets \ABBA[\ltuple \IDMVBA, \groupindex, \electionoutput, \networksize,  \networkfaultsize\rtuple](\ABBAoutput)$        \label{line:OciorRMVBAABAoutput} 		          \quad  \quad    \quad  \quad\quad   \quad  \quad\quad     \emph{//   asynchronous  binary  BA,  operated within $\Sc_{\groupindex}$, with $\networksize, \networkfaultsize$ parameters  }  
 
			\If {$\ABAoutput \eqlog 1$} 
				\State  $\Pass$ $\ABAoutput$ into  $\OciorRBA[(\IDMVBA, \groupindex, \electionoutput)]$ as a binary input   (other than the   message input)     
				\State $\wait$ for $\OciorRBA[(\IDMVBA, \groupindex, \electionoutput)]$ to output value $\MVBAOutputMsg$	        \label{line:OciorRMVBAOciorRBAoutput}  
					\If {$\Predicate(\MVBAOutputMsg) \eqlog \true$}  \label{line:OciorRMVBAEvalidity} 	
						\State  $\Output$  $\MVBAOutputMsg$ and $\terminate$ 	this  $\RecursiveMVBA[(\IDMVBA, \groupindex)]$ and all invoked recursive protocols  under it.	  \label{line:OciorRMVBALoopEnd} 			  	 			
					\EndIf 			
			\EndIf
		\EndFor		
	\EndIf
\EndIndent

\Statex

\State {\bf upon} $\RecursiveMVBA[(\IDMVBA, 2\groupindex+\setindex)]$  outputs $\MVBAOutputMsg$, with $\Predicate(\MVBAOutputMsg) \eqlog \true$,   and $\Node_{\thisnodeindex} \in \Sc_{2\groupindex+\setindex}$ {\bf do}:  	      \label{line:OciorRMVBAInputSHMDACond} 	   
\Indent
	\State  $\Pass$  $\MVBAOutputMsg$ into $\SHMDM[(\IDMVBA, \groupindex, \setindex)]$ as a message input       \label{line:OciorRMVBAInputSHMDA} 	
\EndIndent

\State {\bf upon} $\SHMDM[(\IDMVBA, \groupindex, \setindexnew)]$ outputs $\MVBAOutputMsg$, with $\Predicate(\MVBAOutputMsg) \eqlog \true$,  for $\setindexnew\in \{0,1\}$  and $\Node_{\thisnodeindex} \in \Sc_\groupindex$ {\bf do}:      \label{line:OciorRMVBAInputRBACond} 
\Indent
	\State $\Pass$  $\MVBAOutputMsg$ into $\OciorRBA[(\IDMVBA, \groupindex, \setindexnew)]$  as a message input   \quad  \emph{// $\OciorRBA[(\IDMVBA, \groupindex, \setindexnew)]$ is  a reliable BA protocol  operated within  $\Sc_{\groupindex}$   }       \label{line:OciorRMVBAInputRBA} 
\EndIndent

\Statex

\State {\bf upon} $\OciorRBA[(\IDMVBA, \groupindex, \setindexnew)]$ delivers  $\Vr_i =1$,  for $\setindexnew\in \{0,1\}$  and $\Node_{\thisnodeindex} \in \Sc_\groupindex$ {\bf do}:  	    \label{line:RMVBAReadyindicatorCondition} 
\Indent
	\State $\DRBCReadyindicator[\setindexnew] \gets 1$      \quad\quad\quad\quad \quad\quad\quad\quad\quad\quad\quad \quad\quad\quad\quad\quad\quad\quad \quad  \emph{//     ready  } 
	\State $\send$ $(\READY,  \IDMVBA,  \Round_\groupindex, \setindexnew)$  to  all nodes within $\Sc_\groupindex$         
\EndIndent

\Statex

\State {\bf upon} receiving   $\networksize-\networkfaultsize$  $(\READY, \IDMVBA,  \Round_\groupindex, \setindexnew)$  messages from distinct nodes within $\Sc_\groupindex$   for $\setindexnew\in \{0,1\}$  and $\Node_{\thisnodeindex} \in \Sc_\groupindex$ {\bf do}:      \label{line:RMVBAFinishindicatorCondition} 
\Indent  
	\State $\DRBCFinishindicator[\setindexnew] \gets 1$     \quad\quad\quad\quad \quad\quad\quad\quad\quad\quad\quad \quad\quad\quad\quad\quad\quad\quad \quad  \emph{//     finish  } 
	\State $\send$ $(\FINISH, \IDMVBA,  \Round_\groupindex, \setindexnew)$ to  all nodes within $\Sc_\groupindex$      	 
\EndIndent

\Statex

\State {\bf upon} receiving   $\networksize-\networkfaultsize$  $(\FINISH, \IDMVBA,  \Round_\groupindex, \setindexnew)$  messages from  distinct nodes within $\Sc_\groupindex$  for $\setindexnew\in \{0,1\}$  and $\Node_{\thisnodeindex} \in \Sc_\groupindex$ {\bf do}:  
\Indent  
	\State $\DRBCConfirmindicator[\setindexnew] \gets 1$     \quad\quad\quad \quad\quad\quad\quad\quad\quad\quad \quad\quad\quad\quad\quad\quad\quad \quad  \emph{//     confirm  } 
\EndIndent

\EndProcedure

\end{algorithmic}
\end{algorithm}

\begin{algorithm}[H]
\caption{$\ABBBA$  protocol with  identifier $\ltuple \IDMVBA, \groupindex,  \electionoutput, \networksize,  \networkfaultsize\rtuple$,  for $\IDMVBAtilde:=\ltuple \IDMVBA, \floor{\log \groupindex} +1,  \electionoutput  \rtuple $. This protocol operates on a network $\Sc_{\groupindex}$ of $\networksize$ nodes, up to $\networkfaultsize$ of which may be dishonest.    Code is shown for   $\Node_{\thisnodeindex}$.}    \label{algm:ABBBA} 
\begin{algorithmic}[1]
\vspace{5pt}    
\footnotesize

\Statex   \emph{//   ** Each node inputs a pair of numbers $(\abbainputA, \abbainputB)$, for some $\abbainputA, \abbainputB \in \{0,1\}$ **}   

\Statex   \emph{//   ** Terminate is guaranteed if  the following condition is satisfied:  if one honest node inputs $\abbainputB =1$, then  at least $\networkfaultsize+1$ honest nodes  input  $\abbainputA =1$ **}    

\Statex   \emph{//   ** If at least $\networkfaultsize+1$ honest nodes input $\abbainputB=1$, then none of the honest nodes will  output $0$. This is because  at most $\networksize-(\networkfaultsize+1)$ nodes input $\abbainputB=0$ in this case, indicating that the condition in Line~\ref{line:zerocondition} could not be satisfied. **}    

\Statex   \emph{//   ** If one honest node outputs $1$ (only when $(\ABBACountA \geq \networkfaultsize+1) \OR (\ABBACountB \geq \networkfaultsize+1)$), then at least one honest node has an input as $\abbainputA =1$ or  $\abbainputB =1$ **}

\Statex

\State {\bf upon} receiving  input  $(\abbainputA, \abbainputB)$, for some $\abbainputA, \abbainputB \in \{0,1\}$   {\bf do}:  		
\Indent
	\State $\ABBACountA \gets 0; \ABBACountB \gets 0; \ABBACountC \gets 0$   
	\State $\send$   $(\ABBAVALUE,  \IDMVBAtilde, \abbainputA,  \abbainputB)$   to all nodes
	
	\If {$(\abbainputA \eqlog 1) \OR (\abbainputB \eqlog 1)$ } $\Output$  $1$ and $\terminate$ 
	\EndIf

	\State {\bf wait} for  at least one of the following events: 1) $\ABBACountA \geq \networkfaultsize+1$,  2) $\ABBACountB \geq \networkfaultsize+1$,  or 3)  $\ABBACountC \geq \networksize-\networkfaultsize$   
	\Indent
		\If {$(\ABBACountA \geq \networkfaultsize+1) \OR (\ABBACountB \geq \networkfaultsize+1)$}		 
			\State  $\Output$  $1$ and $\terminate$        \label{line:ABBBAoneoutput} 
		\ElsIf {$\ABBACountC \geq \networksize-\networkfaultsize$}		     \label{line:zerocondition} 
			\State  $\Output$  $0$ and $\terminate$ 
		\EndIf
	\EndIndent	
				
\EndIndent

\State {\bf upon} receiving  $(\ABBAVALUE, \IDMVBAtilde, \abbainputA,  \abbainputB)$ from  $\Node_j$ for the first time, for some $\abbainputA, \abbainputB \in \{0,1\}$ {\bf do}:  
\Indent  
	\State  $\ABBACountA \gets \ABBACountA +\abbainputA; \ABBACountB \gets \ABBACountB +\abbainputB$   
	\If {$\abbainputB \eqlog 0$ } $\ABBACountC \gets \ABBACountC +1$
	\EndIf			

\EndIndent

\end{algorithmic}
\end{algorithm}

\begin{algorithm}[H]
\caption{$\SHMDM$ protocol with identifier $(\IDMVBA, \groupindex, \setindex)$,  for $\IDMVBAtilde:=(\IDMVBA, \floor{\log \groupindex} +1, \setindex)$, and for $\setindex\in \{0,1\}$.  Code is shown for    $\Node_{\thisnodeindex}$, where  $\Node_{\thisnodeindex}$ denotes the $i$th node within  $\Sc_{\groupindex}$. }  \label{algm:SHMDM}
\begin{algorithmic}[1]
\vspace{5pt}    
 
\footnotesize

 \Statex   \emph{//   **  $\SHMDM$:   Strongly-honest-majority distributed multicast   **}

\Statex   \emph{//   **  $\Sc_\groupindex$ is partitioned into two disjoint sets $\Sc_{2\groupindex}$ and  $\Sc_{2\groupindex+1}$ with   $|\Sc_{2\groupindex}| =\floor{|\Sc_{\groupindex}|/2}$ and  $|\Sc_{2\groupindex+1}| =\ceil{|\Sc_{\groupindex}|/2}$**}

\State  Initially set   $\OECsymbolset \gets  \{\}, \networksizestar \gets  |\Sc_{2\groupindex+\setindex}|;  \networkfaultsizestar \gets \floor{\frac{|\Sc_{2\groupindex+\setindex}|-1}{3}};  \kstar \gets  \networkfaultsizestar+1$
		
\State {\bf upon} receiving input   $\MVBAInputMsg$  {\bf do}:      
\Indent  
 
    	\If {$\Node_{\thisnodeindex} \in \Sc_{2\groupindex+\setindex}$}   
		\State $\nodeindexstar\gets  i- \setindex\cdot|\Sc_{2\groupindex}|$               \quad\quad\quad\quad  \quad \quad\quad\quad\quad\quad\quad \quad\quad\quad\quad\quad  \emph{//      $\nodeindexstar$ is the position of this node within  $\Sc_{2\groupindex+\setindex}$   }  
		\State  $z_{\nodeindexstar} \gets \ECCEnc_{\nodeindexstar}(\networksizestar,  \kstar, \MVBAInputMsg)$     \quad \quad \quad\quad\quad\quad\quad\quad\quad\quad\quad  \emph{//     $\ECCEnc_{\nodeindexstar}$ outputs the $\nodeindexstar$th encoded symbol only  }
		\State   $\send$ $\ltuple \INITIAL, \IDMVBAtilde, z_{\nodeindexstar}\rtuple$ to all nodes in $\Sc_{\groupindex}\setminus \Sc_{2\groupindex+\setindex}$           \quad \emph{//     broadcast coded symbol to other set for decoding initial message    }  
		\State   $\Output$   $\MVBAInputMsg$ and $\terminate$     
       \EndIf	  
\EndIndent

\State {\bf upon} receiving   $\ltuple\INITIAL, \IDMVBAtilde, z \rtuple$  from  $\Node_j$ for the first time for some $z$, for $\Node_j \in   \Sc_{2\groupindex+\setindex}$, and $\Node_{\thisnodeindex} \in \Sc_{\groupindex}\setminus \Sc_{2\groupindex+\setindex}$    {\bf do}:  
\Indent  
		\State $j^{\star}\gets  j- \setindex\cdot|\Sc_{2\groupindex}|; \quad \OECsymbolset[j^{\star}] \gets z$         \quad  \quad\quad\quad \quad\quad\quad \quad\quad\quad  \emph{//       $j^{\star}$ is the position of $\Node_j$ within  $\Sc_{2\groupindex+\setindex}$   }  
 		\If  { $|\OECsymbolsetInitial|\geq  \kstar + \networkfaultsizestar  $}    \quad \quad \quad \quad\quad\quad \quad   \quad\quad \quad \quad\quad\quad \quad    \emph{//   online error correcting  (OEC)  }  
			\State   $\tilde{\wv}  \gets \ECCDec(\networksizestar,  \kstar, \OECsymbolsetInitial)$	
			\State  $ [\OECsymbol_{1}', \OECsymbol_{2}', \cdots, \OECsymbol_{\networksizestar}'] \gets \ECCEnc (n,  k, \tilde{\wv})$ 
    			\If {at least $\kstar + \networkfaultsizestar$ symbols in $ [\OECsymbol_{1}', \OECsymbol_{2}', \cdots, \OECsymbol_{\networksizestar}'] $ match with  those in $\OECsymbolsetInitial$}
 				\State   $\Output$   $\tilde{\wv}$ and $\terminate$          			
    			\EndIf    
		\EndIf   
	    		
\EndIndent

\end{algorithmic}
\end{algorithm}

\begin{algorithm}[H]
\caption{$\OciorRBA$ protocol with identifier $(\IDCOOL, \groupindex, \setindex)$,  for $\IDMVBAtilde:=(\IDCOOL, \floor{\log \groupindex} +1, \setindex)$.   Code is shown for    $\Node_{\thisnodeindex}$, where  $\Node_{\thisnodeindex}$ denotes the $i$th node within  $\Sc_{\groupindex}$. 
 This protocol is operated within   $\Sc_{\groupindex}$.  }  \label{algm:OciorRBA}
\begin{algorithmic}[1]
\vspace{5pt}     
 
\footnotesize
 \Statex   \emph{//   **   $\OciorRBA$  is an error-free  reliable Byzantine agreement ($\RBA$) protocol,  extended from $\OciorCOOL$ and $\OciorRBC$ \cite{ChenOciorCOOL:24}. **}

\State  Initially set   $\networksize \gets  |\Sc_{\groupindex}|;  \networkfaultsize\gets \floor{\frac{|\Sc_{\groupindex}|-1}{3}}$,   $\ktilde\gets \bigl \lfloor   \frac{ \networkfaultsize  }{5 } \bigr\rfloor    +1$; $\OECSIFinal\gets 0;   \OECCorrectSymbolSet\gets \{\}; \Lkset_0\gets \{\}; \Lkset_1\gets \{\}; \Ss_0^{[1]}\gets \{\};\Ss_1^{[1]}\gets \{\};\Ss_0^{[2]}\gets \{\};\Ss_1^{[2]}\gets \{\};     \ECCEncindicator \gets 0; \SIPhtwo\gets 0;   \Phtwoindicator\gets 0; \Phthreeindicator\gets 0$

\Statex {\bf \emph{Phase~1}}       

\State {\bf upon} receiving a non-empty  message  input $\Me_{i}$  {\bf do}:   
\Indent  
 \State  $\Me^{(i)} \gets \Me_{i}$   	  \label{line:RBAph0B} 	
 \State    $[y_1^{(i)}, y_2^{(i)}, \cdots, y_{\networksize}^{(i)}]\gets \ECCEnc(\networksize, \ktilde, \Me_{i})$     \label{line:RBAECCEnc}
 \State   $\send$  $\ltuple   \SYMBOL, \IDMVBAtilde,  (y_j^{(i)}, y_i^{(i)}) \rtuple$ to $\Node_j$,  $\forall j \in  [1:\networksize]$, and then set $\ECCEncindicator\gets 1$    \quad\quad \quad\quad \quad\quad \quad\quad   \quad  \ \emph{// exchange coded  symbols }   		 \label{line:RBAECCEncindicator}
  
\EndIndent

\State {\bf upon} receiving   $\ltuple\SYMBOL, \IDMVBAtilde, (y_i^{(j)}, y_j^{(j)}) \rtuple$  from  $\Node_j$ for the first time  {\bf do}:  
\Indent  
 
	\State  $\wait$ until  $\ECCEncindicator =1$ 
	\If {$ (y_i^{(j)}, y_j^{(j)}) = (y_i^{(i)}, y_j^{(i)})$ }         \label{line:RBAph1MatchCond}
		\State  $\Lkset_1\gets \Lkset_1\cup \{j\}$              \label{line:RBAph1Match}  \quad\quad\quad \quad\quad\quad\quad  \quad \quad \quad\quad \quad \quad \quad \quad \quad\quad \quad \quad  \quad \quad \quad \quad \quad \quad \quad \quad  \quad \quad \quad \quad \emph{// update the set of link indicators}
	\Else 
		\State  $\Lkset_0\gets \Lkset_0\cup \{j\}$            \label{line:RBAph1NotMatch} 
	\EndIf
 
\EndIndent

\State {\bf upon}  $|\Lkset_1|\geq  \networksize- \networkfaultsize $,  and  $\ltuple\SIone, \IDMVBAtilde, * \rtuple$  not yet sent {\bf do}:      \label{line:RBAph2OneCond}
\Indent  
           \State	set  $\Ry_i^{[1]} \gets 1$,  $\send$ $\ltuple\SIone, \IDMVBAtilde, \Ry_i^{[1]}\rtuple$ to all nodes, and then set $\Phtwoindicator\gets 1$    \label{line:RBAph2SI1} \quad \quad\quad\quad\quad\quad \quad\quad\quad\quad\quad\quad \quad   \quad\quad\quad \emph{// set success indicator}
 
\EndIndent

\State {\bf upon}  $|\Lkset_0|\geq   \networkfaultsize +1$, and  $\ltuple\SIone, \IDMVBAtilde, * \rtuple$  not yet sent {\bf do}:       \label{line:RBAph2ZeroCond}
\Indent  
           \State	set $\Ry_i^{[1]} \gets 0$,   $\send$ $\ltuple\SIone, \IDMVBAtilde, \Ry_i^{[1]}\rtuple$ to all nodes, 	and  then set $\Phtwoindicator\gets 1$    \label{line:RBAph2B}
\EndIndent

\State {\bf upon} receiving   $\ltuple\SIone, \IDMVBAtilde, \Ry_j^{[1]}\rtuple$  from  $\Node_j$ for the first time   {\bf do}:    \label{line:RBAph1SS01Cond}  
\Indent  
	\If {$ \Ry_j^{[1]} =1$ }     
		\State  $\wait$ until  $(j \in  \Lkset_1\cup \Lkset_0)\OR(|\Ss_1^{[1]}|\geq  \networksize- \networkfaultsize)\OR(|\Ss_0^{[1]}|\geq   \networkfaultsize +1)$ 
		\If {$ j \in  \Lkset_1$ }         \label{line:RBAph1SS1Cond} 
			\State  $\Ss_1^{[1]}\gets \Ss_1^{[1]}\cup \{j\}$        \label{line:RBAph1SS1}         \quad  \quad \quad \quad\quad \quad \quad \quad \quad \quad\quad \quad \quad  \quad \quad \quad \quad \quad \quad \quad \quad  \quad \quad \quad \quad \emph{//update the set of success indicator as ones}
		\ElsIf{$ j \in  \Lkset_0$ }
			\State  $\Ss_0^{[1]} \gets \Ss_0^{[1]}\cup \{j\}$       \label{line:RBAph1SS0}   \quad  \quad  \quad \quad \quad\quad \quad \quad \quad \quad \quad\quad \quad \quad  \quad \quad \quad \quad \quad \quad \quad \quad  \quad \quad \quad \quad \emph{// mask identified errors  (mismatched links)}
		\EndIf
	\Else 
		\State  $\Ss_0^{[1]} \gets \Ss_0^{[1]}\cup \{j\}$        \label{line:RBAph1SS0NotM}  \quad\quad  \quad  \quad \quad \quad\quad \quad \quad \quad \quad \quad\quad \quad \quad  \quad \quad \quad \quad \quad \quad \quad \quad  \quad \quad \quad \quad \emph{// mask identified errors  (mismatched links)}
	\EndIf
 
\EndIndent

\Statex  {\bf \emph{Phase~2}}    

 \State {\bf upon}  $(\Phtwoindicator= 1)\AND(\Ry_i^{[1]} = 0)$, and  $\ltuple\SItwo, \IDMVBAtilde, \Ry_i^{[2]}\rtuple$ not yet sent {\bf do}:  
\Indent  

	\State	set $\Ry_i^{[2]} \gets 0$,   $\send$ $\ltuple\SItwo, \IDMVBAtilde, \Ry_i^{[2]}\rtuple$ to all nodes  \label{line:RBAph2CSI0} \quad\quad \quad\quad\quad\quad\quad\quad  \quad\quad\quad\quad \quad \quad\quad\quad\quad \quad \quad\quad\quad\quad \quad   \emph{// update success indicator }
	
\EndIndent

 \State {\bf upon}  $(\Phtwoindicator= 1)\AND(\Ry_i^{[1]} = 1) \AND(|\Ss_1^{[1]}|\geq  \networksize- \networkfaultsize)$, and  $\ltuple\SItwo, \IDMVBAtilde, \Ry_i^{[2]}\rtuple$ not yet sent   {\bf do}:   \label{line:RBAph2ACond} 
\Indent   

           \State	set  $\Ry_i^{[2]} \gets 1$, $\SIPhtwo\gets 1$, and   $\send$ $\ltuple \SItwo, \IDMVBAtilde,\Ry_i^{[2]}\rtuple$ to all nodes   \label{line:RBAph2ASI1} 
	
\EndIndent

 \State {\bf upon}  $|\Ss_0^{[1]}|\geq   \networkfaultsize +1$, and  $\ltuple\SItwo, \IDMVBAtilde, \Ry_i^{[2]}\rtuple$ not yet sent   {\bf do}:  
\Indent  

           \State	set $\Ry_i^{[2]} \gets 0$,   $\send$ $\ltuple \SItwo, \IDMVBAtilde, \Ry_i^{[2]}\rtuple$ to all nodes    \label{line:RBAph2BSI0}     
	
\EndIndent

 \State {\bf upon} receiving   $\ltuple\SItwo, \IDMVBAtilde, \Ry_j^{[2]}\rtuple$  from  $\Node_j$ for the first time  {\bf do}:     \label{line:RBAph2SS01Cond}     
\Indent  
	\If {$ \Ry_j^{[2]} =1$ }     
		\State  $\wait$ until $(j \in  \Lkset_1\cup \Lkset_0)\OR(|\Ss_1^{[2]}|\geq  \networksize- \networkfaultsize)\OR(|\Ss_0^{[2]}|\geq   \networkfaultsize +1)$ 
		\If {$ j \in  \Lkset_1$ }     
			\State  $\Ss_1^{[2]}\gets \Ss_1^{[2]}\cup \{j\}$          
		\ElsIf{$ j \in  \Lkset_0$ }
			\State  $\Ss_0^{[2]} \gets \Ss_0^{[2]}\cup \{j\}$            
		\EndIf
	\Else 
		\State  $\Ss_0^{[2]} \gets \Ss_0^{[2]}\cup \{j\}$            \label{line:RBAph2SS0NM}     
	\EndIf
 
\EndIndent

\algstore{OciorRBA}

\end{algorithmic}
\end{algorithm}

\begin{algorithm}[H]
\begin{algorithmic}[1]
\algrestore{OciorRBA}
\vspace{5pt}    
 \footnotesize

\State {\bf upon}   $|\Ss_{\Vr}^{[2]}|\geq  \networksize- \networkfaultsize$,     for a  $\Vr \in \{1,0\}$,  and $\ltuple\READY, \IDMVBAtilde, * \rtuple$  not yet sent {\bf do}:      \label{line:RBAReadyCondition}   
\Indent  
		\State  set $\Vr_i \gets \Vr$, and  $\deliver$  $\Vr_i$        \label{line:RBAVrideliver}              \quad \quad\quad \quad  \quad\quad\quad    \quad\quad\quad \quad   \quad\quad\quad\quad  \emph{//  the value of  $\Vr_i$ is  delivered out to the protocol in Algorithm~\ref{algm:OciorRMVBA}}       
		\State $\send$ $\ltuple\READY, \IDMVBAtilde, \Vr_i \rtuple$ to  all nodes  	  \label{line:RBAReadySendA}
\EndIndent

\State {\bf upon} receiving   $\networkfaultsize+1$  $\ltuple \READY, \IDMVBAtilde, \Vr  \rtuple$ messages  from different   nodes for the same $\Vr$ and $\ltuple\READY, \IDMVBAtilde, * \rtuple$  not yet sent {\bf do}:        \label{line:RBARelialbeBegin}   
\Indent  
		\State $\send$ $\ltuple\READY, \IDMVBAtilde, \Vr \rtuple$ to  all nodes	   \label{line:RBARelialbeBeginBB}   
\EndIndent

\State {\bf upon} receiving   $2 \networkfaultsize+1$  $\ltuple \READY, \IDMVBAtilde, \Vr  \rtuple$ messages  from different nodes   for the same $\Vr$ {\bf do}:    \label{line:RBARelialbeEnd}  
\Indent  
	\If {$\ltuple\READY, \IDMVBAtilde, * \rtuple$  not yet sent }     
		\State $\send$ $\ltuple\READY, \IDMVBAtilde, \Vr \rtuple$ to  all nodes    \label{line:RBAReadySendC}
	\EndIf
	\State  set $\VrOutput\gets \Vr$  \label{line:RBAVrOutput}       
	\If {$\VrOutput =0$ }     
		\State set $\Me^{(i)}\gets \defaultvalue$, then $\Output$   $\Me^{(i)}$ and $\terminate$  	 \label{line:RBAoutputdefault}    \quad \quad	\quad\quad\quad\quad\quad\quad\quad\quad\quad\quad \quad\quad\quad\quad   \quad\quad\quad\quad \quad\quad     \emph{//  $\defaultvalue$ is a default value  }    
	\Else
		\State  set $\Phthreeindicator\gets 1$        \label{line:RBAph3}     
		
	\EndIf
\EndIndent

 \State {\bf upon} receiving a binary  input $\ABAoutput =1$  (other than the message  input $\Me_{i}$) {\bf do}:   \label{line:RBAph3BBCondition} \quad\quad \quad  \emph{//   $\ABAoutput$ is  delivered from the protocol in Algorithm~\ref{algm:OciorRMVBA}}       
\Indent  
		\State  set $\Phthreeindicator\gets 1$        \label{line:RBAph3BB}     
\EndIndent

\Statex

\Statex  {\bf \emph{Phase~3}}

\State {\bf upon} $\Phthreeindicator= 1$ {\bf do}:     \label{line:RBAph3Begin}  \quad\quad\quad\quad\quad\quad\quad\quad \quad \quad\quad\quad\quad\quad  \quad \quad\quad \quad \quad\quad\quad\quad\quad  \quad \quad\quad\quad\quad\quad\quad \quad      \emph{//     only after executing Line~\ref{line:RBAph3} or Line~\ref{line:RBAph3BB}}
\Indent  
\If { $\SIPhtwo = 1$}          \label{line:RBAph3SIPhtwo}
	\State $\Output$   $\Me^{(i)}$ and $\terminate$     \label{line:RBAph3SIPhtwoOutput}
		 
\Else
 
			\State  $\wait$ until receiving $\networkfaultsize +1$ $\ltuple\SYMBOL, \IDMVBAtilde, (y_i^{(j)},  *) \rtuple$  messages, $\forall j \in   \Ss_1^{[2]}$, for the same    $y_i^{(j)} =y^{\star}$,  for some   $y^{\star}$   \label{line:RBAph3MajorityRuleCond}
			\State  $y_i^{(i)} \gets y^{\star}$   \label{line:RBAph3MajorityRule}   \quad\quad\quad\quad\quad\quad \quad\quad\quad\quad\quad\quad\quad\quad\quad\quad\quad\quad   \quad\quad\quad\quad\quad\quad\quad\quad \quad\quad \emph{// update coded symbol based on  majority rule}  
	\State   $\send$   $\ltuple \CORRECTSYMBOL, \IDMVBAtilde, y_i^{(i)}  \rtuple$  to  all nodes      \label{line:RBAph3SendCorrectSymbols} 
			\State  $\wait$ until  $\OECSIFinal=1$ 	
			\State $\Output$   $\Me^{(i)}$ and $\terminate$     \label{line:RBAph3Output2}	
	
\EndIf

\EndIndent

\State {\bf upon} receiving   $\ltuple\CORRECTSYMBOL, \IDMVBAtilde, y_j^{(j)} \rtuple$  from  $\Node_j$ for the first time,   $j\notin \OECCorrectSymbolSet$, and $\OECSIFinal=0$   {\bf do}:  \label{line:Ph3OECCond}
\Indent  
	\State $\OECCorrectSymbolSet[j] \gets y_j^{(j)}$    
 	\If  { $|\OECCorrectSymbolSet|\geq  \ktilde + \networkfaultsize  $}     \label{line:OECbegin}   \quad \quad    \quad\quad \quad \quad\quad\quad \quad    \quad\quad \quad \quad\quad\quad \quad   \quad\quad \quad \quad\quad\quad \quad   \quad\quad \quad \quad\quad\quad \quad    \emph{//   online error correcting  (OEC)  }  
			\State   $\MVBAOutputMsg  \gets \ECCDec(\networksize,  \ktilde , \OECCorrectSymbolSet)$	     
			\State  $[y_{1}, y_{2}, \cdots, y_{\networksize}] \gets \ECCEnc (\networksize,  \ktilde, \MVBAOutputMsg)$ 
    			\If {at least $\ktilde + \networkfaultsize$ symbols in $[y_{1}, y_{2}, \cdots, y_{\networksize}]$ match with  those in $\OECCorrectSymbolSet$}
 				\State  $ \Me^{(i)} \gets \MVBAOutputMsg; \OECSIFinal\gets 1$     \label{line:OECend}	 
    			\EndIf    

	\EndIf   
		     
\EndIndent

\State {\bf upon} having received   both $\ltuple\SYMBOL, \IDMVBAtilde, (y_i^{(j)}, y_j^{(j)})\rtuple$ and   $\ltuple \SItwo, \IDMVBAtilde, 1\rtuple$ messages from  $\Node_j$, and $j\notin \OECCorrectSymbolSet$, and  $\OECSIFinal=0$   {\bf do}:    \label{line:RBAPh3OECSecCond}
\Indent  
	\State $\OECCorrectSymbolSet[j] \gets y_j^{(j)}$     
	\State run the OEC steps as in Lines~\ref{line:OECbegin}-\ref{line:OECend}      \label{line:Ph3OECAllEnd}
\EndIndent

\end{algorithmic}
\end{algorithm}

\section{$\OciorRMVBA$}\label{sec:OciorRMVBA}  
  
This proposed $\OciorRMVBA$ is an error-free,  information-theoretically secure asynchronous $\MVBA$  protocol.  
$\OciorRMVBA$ doesn't rely on any   cryptographic assumptions, such as  signatures or hashing,  except for  the common coin assumption. 
The design of $\OciorRMVBA$  in this $\MVBA$ setting builds on the protocols   $\COOL$ and $\OciorCOOL$   \cite{Chen:2020arxiv, ChenDISC:21,ChenOciorCOOL:24}.

\subsection{Overview of the  proposed $\OciorRMVBA$ protocol}  

The proposed  $\OciorRMVBA$ is described in Algorithm~\ref{algm:OciorRMVBA}, along with   Algorithms~\ref{algm:ABBBA}-\ref{algm:OciorRBA}. Fig.~\ref{fig:OciorRMVBA} presents a block diagram of the proposed $\OciorRMVBA$ protocol. 
For a network $\Sc_{\groupindex}$, we define the network size and the  faulty threshold   as $\networksize_{\groupindex}:=  |\Sc_{\groupindex}|$ and  $\networkfaultsize_{\groupindex}:= \floor{\frac{|\Sc_{\groupindex}|-1}{3}}$, for $\groupindex\in \{1,2,3,\cdots\}$.  The original network is defined as $\Sc_{1}:=\{\Node_i: i\in [1:n]\}$, where $\Node_i$ denotes the $i$th node in the network.   $\Sc_\groupindex$ is partitioned into two disjoint sets $\Sc_{2\groupindex}$ and  $\Sc_{2\groupindex+1}$ such that  $|\Sc_{2\groupindex}| =\floor{|\Sc_{\groupindex}|/2}$ and  $|\Sc_{2\groupindex+1}| =\ceil{|\Sc_{\groupindex}|/2}$.
Below, we provide an overview of the proposed protocol.

$\OciorRMVBA$ is a recursive protocol. As shown in  Fig.~\ref{fig:OciorRMVBA}, the protocol $\RecursiveMVBA[(\IDMVBA, \groupindex)]$ invokes two sub-protocols: $\RecursiveMVBA[(\IDMVBA, 2\groupindex)] $  and $\RecursiveMVBA[(\IDMVBA, 2\groupindex+1)] $. Each  sub-protocol, in turn,  invokes two additional sub-protocols. This process continues until the size of network on which a sub-protocol operates on is smaller than a predefined finite threshold. In the final protocol invoked, any inefficient $\MVBA$ protocol (referred to as $\IEMVBA$) can be used without impacting overall performance, as the size of the operated network is finite.

The general protocol $\RecursiveMVBA[(\IDMVBA, \groupindex)]$ comprises several steps, as illustrated in Fig.~\ref{fig:OciorRMVBA}.
It operates over the network $\Sc_\groupindex$, which is partitioned into two disjoint sets, $\Sc_{2\groupindex}$ and $\Sc_{2\groupindex+1}$, of balanced size. The two invoked protocols operate on these partitioned network sets.  The key steps involved in $\RecursiveMVBA[(\IDMVBA, \groupindex)]$ are described below. 
\begin{itemize}
\item   \emph{Step 1:}  Upon   an invoked sub-protocol $\RecursiveMVBA[(\IDMVBA, 2\groupindex+\setindex)]$, for $\setindex\in \{0,1\}$,  outputs a message $\Me_{\setindex}$,   the nodes within   $\Sc_{2\groupindex+\setindex}$  input  $\Me_{\setindex}$ into  $\SHMDM[(\IDMVBA, \groupindex, \setindex)]$.  
\item   \emph{Step 2:}  Upon   $\SHMDM[(\IDMVBA, \groupindex, \setindex)]$ outputs a message $\Me_{\setindex}$,   the nodes within   $\Sc_{\groupindex}$  input  $\Me_{\setindex}$ into  $\OciorRBA[(\IDMVBA, \groupindex, \setindex)]$.      
\item   \emph{Step 3:} Upon   $\OciorRBA[(\IDMVBA, \groupindex, \setindex)]$ delivers  $\Vr_i =1$,     the nodes within   $\Sc_{\groupindex}$  set $\DRBCReadyindicator[\setindex] \gets 1$   and send    $(\READY,  \IDMVBA,  \Round_\groupindex, \setindex)$  to  all nodes within $\Sc_\groupindex$, where $\Round_\groupindex:= \floor{\log \groupindex} +1$.
\item   \emph{Step 4:} Upon   receiving   $\networksize_{\groupindex}-\networkfaultsize_{\groupindex}$  $(\READY, \IDMVBA,  \Round_\groupindex, \setindex)$  messages from distinct nodes within $\Sc_\groupindex$,   the nodes within   $\Sc_{\groupindex}$    set $\DRBCFinishindicator[\setindex] \gets 1$     and send   $(\FINISH, \IDMVBA,  \Round_\groupindex, \setindex)$ to  all nodes within $\Sc_\groupindex$ .
\item   \emph{Step 5:} Upon   receiving   $\networksize_{\groupindex}-\networkfaultsize_{\groupindex}$   $(\FINISH, \IDMVBA,  \Round_\groupindex, \setindex)$  messages from  distinct nodes within $\Sc_\groupindex$,   the nodes within   $\Sc_{\groupindex}$    set $\DRBCConfirmindicator[\setindex] \gets 1$.
\item   \emph{Step 6:} After setting $\DRBCConfirmindicator[\setindex] \gets 1$,   the nodes within   $\Sc_{\groupindex}$ proceed to the Ordered Election  step, which outputs   $\electionoutput$,   taking a value from  $  \{0, 1\}$ in order.
\item   \emph{Step 7:} The nodes within   $\Sc_{\groupindex}$   run the $\ABBBA$ protocol with inputs $(\DRBCReadyindicator[\electionoutput], \DRBCFinishindicator[\electionoutput])$.
\item   \emph{Step 8:} After  $\ABBBA$ outputs a value $\ABBAoutput$,  the nodes within   $\Sc_{\groupindex}$   run an asynchronous binary $\BA$ ($\ABBA$) protocol with input $\ABBAoutput$. 
\item   \emph{Step 9:} If  $\ABBA$ outputs a value $1$,  the nodes within   $\Sc_{\groupindex}$ input $1$ into  $\OciorRBA[(\IDMVBA, \groupindex, \electionoutput)]$ and wait for its output.  If $\OciorRBA[(\IDMVBA, \groupindex, \electionoutput)]$ outputs a value $\Me_{\electionoutput}$ such that $\Predicate(\Me_{\electionoutput}) \eqlog \true$,    the nodes within   $\Sc_{\groupindex}$ output  $\Me_{\electionoutput}$ and terminate the protocol  $\RecursiveMVBA[(\IDMVBA, \groupindex)]$.   If  $\ABBA$ outputs a value $0$, the nodes go back to Step~6. 
\end{itemize} 
In our design,  by combining the   the Ready-Finish-Confirm process  in Steps~4-6 with $\ABBBA$ in Steps~7, we ensure  that if one honest node has set $\DRBCConfirmindicator[\electionoutput] \gets 1$, eventually every honest node outputs $1$ from $\ABBBA [\ltuple \IDMVBA, \groupindex,  \electionoutput, \networksize_{\groupindex},  \networkfaultsize_{\groupindex}\rtuple]$ in Step~7,  due to the Biased Validity  property of $\ABBBA$.
It is worth noting that the design of $\OciorRBA$ protocol in the $\RBA$ setting builds upon  the protocols $\COOL$ and $\OciorCOOL$, which utilize  $\BUA$ and $\HMDM$ as building blocks \cite{Chen:2020arxiv, ChenDISC:21,ChenOciorCOOL:24}.

\subsection{Analysis of $\OciorRMVBA$}    \label{sec:AnalysisOciorMVBAstar}

\begin{definition} [\bf Good resilience]    
A network $\Sc_{\groupindex}$ is said to have good resilience if the number of dishonest nodes within $\Sc_{\groupindex}$, denoted by $\networkfaultsizereal_{\groupindex}:=|\Sc_{\groupindex}\cap \Fc|$,  satisfies the condition  $\networkfaultsizereal_{\groupindex}< \frac{|\Sc_{\groupindex}|}{3}$, where $\Fc$ denotes the set of all dishonest nodes.  This condition  $\networkfaultsizereal_{\groupindex}< \frac{|\Sc_{\groupindex}|}{3}$   is equivalent to  $\networkfaultsizereal_{\groupindex}\leq  \frac{|\Sc_{\groupindex}|-1}{3}$, and also equivalent to the condition   $\networkfaultsizereal_{\groupindex}\leq \floor{\frac{|\Sc_{\groupindex}|-1}{3}}$ due to the integer nature of $\networkfaultsizereal_{\groupindex}$. 
\end{definition}

\begin{definition} [\bf Network tree]    
In our setting, $\Sc_\groupindex$ is partitioned into two disjoint sets $\Sc_{2\groupindex}$ and  $\Sc_{2\groupindex+1}$.  
For example, $\Sc_1$  (at Layer 1) is partitioned into two disjoint sets: $\Sc_{2}$ and   $\Sc_{3}$ (at Layer~2).  The two sets  $\Sc_{2}$ and $\Sc_{3}$ are partitioned into $\Sc_{4}$, $\Sc_{5}$ and $\Sc_{6}$, $\Sc_{7}$, respectively, at Layer~3. 
We define the network tree as comprising all layers of sets: $\Sc_1$  at Layer 1;  $\Sc_{2}$ and   $\Sc_{3}$ at Layer~2; $\Sc_{4}$, $\Sc_{5}$, $\Sc_{6}$ and  $\Sc_{7}$ at Layer~3; and so on. 
\end{definition}

\begin{definition} [\bf Network chain]    
By selecting one network set at each layer from a network tree, where the set selected at Layer~$r$ is partitioned from the set selected at Layer~$r-1$, for $r=2,3,\cdots$, then the selected network sets form a network chain.
One example of a network chain is   $\Sc_{1}\to \Sc_{3} \to \Sc_{7} \to \Sc_{15} \to\cdots $.
\end{definition}

\begin{definition} [\bf Network chain with good resilience]    
A network chain is considered to have good resilience if every network set it includes  has good resilience. For example, if all of $\Sc_{1}, \Sc_{3}, \Sc_{7} , \Sc_{15}, \cdots$  have good resilience,  then the  network chain    $\Sc_{1}\to \Sc_{3} \to \Sc_{7} \to \Sc_{15} \to\cdots $ is considered to have good resilience.
\end{definition}

\begin{theorem}  [Agreement and Termination]  \label{thm:RMVBAaTermination}
In $\OciorRMVBA$,  given $n\geq 3t+1$, every honest node eventually outputs a consistent value and terminates.     
\end{theorem}
\begin{proof}
In our setting, $\Sc_\groupindex$ is partitioned into two disjoint sets $\Sc_{2\groupindex}$ and  $\Sc_{2\groupindex+1}$.  
From Lemma~\ref{lm:RMVBAgoodResGeneral},  if  $\Sc_\groupindex$ has good resilience, i.e.,  $|\Sc_{\groupindex} \cap \Fc| < \frac{|\Sc_{\groupindex}|}{3}$,  then at least one of the two sets, $\Sc_{2\groupindex}$ or  $\Sc_{2\groupindex+1}$, also has good resilience. 
From Lemma~\ref{lm:RMVBAgoodResNChain}, given $n\geq 3t+1$, there exists a network chain with good resilience.

 Let us  focus on a  network chain $\Sc_1\to \cdots \to \Sc_\groupindex \to \Sc_{2\groupindex+ \setindex} \to  \cdots$  with good resilience, for some $\setindex \in \{0,1\}$. 
From Lemma~\ref{lm:RMVBAgoodResOutTer}, it is true that  if   $\RecursiveMVBA[(\IDMVBA, 2\groupindex+\setindex)]$  outputs a consistent value at all honest nodes within  $\Sc_{2\groupindex+ \setindex}$ and terminates,   then  $\RecursiveMVBA[(\IDMVBA, \groupindex)]$ eventually outputs a consistent value at all honest nodes within $\Sc_\groupindex$ and terminates. 
Based on a recursive argument, and given that the last invoked $\RecursiveMVBA$  protocol eventually outputs a consistent value at all honest nodes within the last network set in the chain, it is concluded that 
every honest node within  $\Sc_1$ eventually outputs a  consistent value and terminates, given that $n\geq 3t+1$.    
\end{proof}

\begin{theorem}  [External Validity]  \label{thm:RMVBAaExternalvalidity}
In $\OciorRMVBA$,  if an honest node outputs a value $\MVBAInputMsg$, then $\Predicate(\MVBAInputMsg)=\true$.    
\end{theorem}
\begin{proof}
 In $\OciorRMVBA$,  if an honest node outputs a value $\MVBAInputMsg$,  it has verified that $\Predicate(\MVBAInputMsg)=\true$ at Line~\ref{line:OciorRMVBAEvalidity} of Algorithm~\ref{algm:OciorRMVBA}.
\end{proof}

\begin{theorem}  [Communication and Round Complexities]  \label{thm:RMVBAComplexites}
The communication complexity of  $\OciorRMVBA$ is $O( n  |\MVBAInputMsg|\log n   +  n^2 \log \alphabetsize )$ bits,  while the round complexity of  $\OciorRMVBA$ is $O( \log n )$ rounds, given $n\geq 3t+1$.     
\end{theorem}
\begin{proof}
The protocol $\RecursiveMVBA[(\IDMVBA, \groupindex)]$ is operated over $\Sc_\groupindex$ with a network size of $\networksize_{\groupindex}:=  |\Sc_{\groupindex}|$.
The total communication complexity  in bits of the protocol $\RecursiveMVBA[(\IDMVBA, \groupindex)]$, defined by $\CCbits (\networksize_{\groupindex})$ is expressed as 
\begin{equation*}
 \CCbits (\networksize_{\groupindex})  = 
 \begin{cases}
O( |\MVBAInputMsg| )    &\text{if  $m \leq M$}\\
   \beta_1  \networksize_{\groupindex}  |\MVBAInputMsg|  +  \beta_2 \networksize_{\groupindex}^2 \log \alphabetsize  +   \CCbits \bigl(\floor{\frac{\networksize_{\groupindex}}{2}}\bigr)+  \CCbits \bigl(\ceil{\frac{\networksize_{\groupindex}}{2}}\bigr)  &\text{otherwise}
\end{cases}
\end{equation*}
where $M$ is a finite constant;    $\beta_1$ and $\beta_2$ are  finite constants; and $|\MVBAInputMsg|$ denotes the size of each input message. It is worth mentioning that each coded symbol transmitted in $\COOLDRBC$ protocol carries at least $\log \alphabetsize $ bits due to  the  alphabet size   of   error correction code.  
We ignore the cost of the index $\Round_\groupindex = \floor{\log \groupindex} + 1$   in transmitted messages (using $\log \log n$ bits), as it can be redesigned such that the total indexing cost related to $\Round_\groupindex$ becomes negligible compared to the cost of coded symbols. 
When $n=2^J M$ for some $J$, then  the total communication complexity in bits of  the proposed $\OciorRMVBA$ is 
\begin{align*}
 \CCbits (n)   = &   \beta_1  n  |\MVBAInputMsg|  +  \beta_2 n^2 \log \alphabetsize+   2\CCbits \bigl(\frac{n}{2}\bigr)   \\   
           =   & \beta_1  n  |\MVBAInputMsg|  +  \beta_2 n^2  \log \alphabetsize+   2\bigl(\beta_1 \cdot  \frac{n}{2} \cdot |\MVBAInputMsg|  +  \beta_2 \frac{n^2}{4} \log \alphabetsize +   2\CCbits \bigl(\frac{n}{2^2}\bigr)\bigr)\\
           =   & \beta_1  n  |\MVBAInputMsg|  +  \beta_2 n^2 \log \alphabetsize+ \beta_1 n |\MVBAInputMsg|  +  \beta_2 \frac{n^2}{2}  \log \alphabetsize +   2^2 \cdot \CCbits \bigl(\frac{n}{2^2}\bigr) \\
              &  \vdots \\
           =   & \beta_1  n  |\MVBAInputMsg|  +  \beta_2 n^2 \log \alphabetsize + \beta_1 n |\MVBAInputMsg|  +  \beta_2 \frac{n^2}{2}  \log \alphabetsize + \cdots +  \beta_1 n |\MVBAInputMsg|  +  \beta_2 \frac{n^2}{2^{J-1}}  \log \alphabetsize  +    2^J \cdot \CCbits \bigl(\frac{n}{2^J}\bigr) \\
           =   &J  \beta_1  n  |\MVBAInputMsg|  +  \beta_2 n^2 \log \alphabetsize \  \cdot (1 +   \frac{1}{2} + \cdots +   \frac{1}{2^{J-1}}) +    2^J \cdot \CCbits \bigl(M\bigr) \\
          =    & O(   n  |\MVBAInputMsg|\log n    +   n^2 \log \alphabetsize ). 
\end{align*}
 
 The round complexity of  $\OciorRMVBA$ is $O( \log n )$ rounds.      
\end{proof}

\begin{lemma}       \label{lm:RMVBAgoodResGeneral}
 If  $\Sc_\groupindex$ has good resilience, i.e.,  $|\Sc_{\groupindex} \cap \Fc| < \frac{|\Sc_{\groupindex}|}{3}$,  then at least one of the two sets, $\Sc_{2\groupindex}$ or  $\Sc_{2\groupindex+1}$, also has good resilience.
\end{lemma}
\begin{proof}
If  $\Sc_\groupindex$ has good resilience, i.e.,  $|\Sc_{\groupindex} \cap \Fc| < \frac{|\Sc_{\groupindex}|}{3}$,  then    at least one of the following two conditions is satisfied:  $|\Sc_{2\groupindex} \cap \Fc|  < \frac{|\Sc_{2\groupindex}|}{3}$ or $|\Sc_{2\groupindex+1} \cap \Fc|  < \frac{|\Sc_{2\groupindex+1}|}{3}$. 
\end{proof}

\begin{lemma}       \label{lm:RMVBAgoodResNChain}
Given $n\geq 3t+1$, there exists a network chain with good resilience. 
\end{lemma}
\begin{proof}
Given $n\geq 3t+1$, $\Sc_1$ has good resilience.  $\Sc_1$ is partitioned into two disjoint sets, $\Sc_{2}$ and  $\Sc_{3}$.  From Lemma~\ref{lm:RMVBAgoodResGeneral}, at least one of the  sets $\Sc_{2}$ and  $\Sc_{3}$ has good resilience. Let us assume that $\Sc_{3}$ has good resilience and   include  it  into a network chain.  Next, $\Sc_3$ is partitioned into two disjoint sets $\Sc_{6}$ and  $\Sc_{7}$.  Again,  from Lemma~\ref{lm:RMVBAgoodResGeneral}, at least one of the  sets $\Sc_{6}$ and  $\Sc_{7}$ has good resilience.  Let us assume that $\Sc_{7}$ has good resilience and   include  it  into the network chain. By repeating this process, we construct  a network chain  $\Sc_{1}\to \Sc_{3} \to \Sc_{7} \to  \cdots $ such that all network sets it includes have good resilience. This  implies that the selected network chain has good resilience. 
\end{proof}

Without loss of generality, we will focus on the network set $\Sc_\groupindex$, such that all other protocols $\RecursiveMVBA[(\IDMVBA, \groupindex')]$ invoking  $\RecursiveMVBA[(\IDMVBA, \groupindex)]$ haven't outputted a value at any honest node yet, where   the network sets  $\Sc_{\groupindex'}$  and  $\Sc_{\groupindex}$ are within the same network chain with  good resilience and  $\groupindex'<\groupindex$.

\begin{lemma}       \label{lm:RMVBAgoodResOutputValuesALL}
Let us assume that $\Sc_\groupindex$ has good resilience.  If one honest node  within $\Sc_\groupindex$ outputs a value $\Me$ from $\RecursiveMVBA[(\IDMVBA, \groupindex)]$, then  all other honest nodes within $\Sc_\groupindex$ eventually outputs a consistent value $\Me$ from $\RecursiveMVBA[(\IDMVBA, \groupindex)]$. 
\end{lemma}
\begin{proof}
Given that $\Sc_\groupindex$ has good resilience, if one honest node  within $\Sc_\groupindex$ outputs a value $\Me$ from $\RecursiveMVBA[(\IDMVBA, \groupindex)]$, then all hones nodes within  $\Sc_\groupindex$ must have output   $1\gets \ABBA[\ltuple \IDMVBA, \groupindex, \electionoutput, \networksize,  \networkfaultsize\rtuple](\ABBAoutput)$  in Line~\ref{line:OciorRMVBAABAoutput}  of Algorithm~\ref{algm:OciorRMVBA}, for some $\electionoutput \in \{0,1\}$. Then, from Lemma~\ref{lm:RMVBAb1ConsistantOutputs} (Consistency and Totality properties of $\OciorRBA$),  all hones nodes within  $\Sc_\groupindex$  eventually  output the same value $\Me$ from $\OciorRBA[(\IDMVBA, \groupindex, \electionoutput)]$ in Line~\ref{line:OciorRMVBAOciorRBAoutput}  of Algorithm~\ref{algm:OciorRMVBA}.
\end{proof}

\begin{lemma}       \label{lm:RMVBAgoodResOutTer}
Let us assume that $\Sc_1\to \cdots \to \Sc_\groupindex \to \Sc_{2\groupindex+ \setindex} \to  \cdots$  forms a network chain with good resilience, for some $\setindex \in \{0,1\}$.  If   $\RecursiveMVBA[(\IDMVBA, 2\groupindex+\setindex)]$  outputs a consistent value at all honest nodes within  $\Sc_{2\groupindex+ \setindex}$ and terminates,   then  $\RecursiveMVBA[(\IDMVBA, \groupindex)]$ eventually outputs a consistent value at all honest nodes within $\Sc_\groupindex$ and terminates. 
\end{lemma}
\begin{proof}
Assume that $\Sc_\groupindex$ and  $\Sc_{2\groupindex+ \setindex}$  have good resilience. 
If   $\RecursiveMVBA[(\IDMVBA, 2\groupindex+\setindex)]$ outputs a consistent value $\Me$ at all honest nodes within  $\Sc_{2\groupindex+ \setindex}$, then eventually  $\SHMDM[(\IDMVBA, \groupindex, \setindex)]$ outputs a consistent value $\Me$ at all honest nodes within  $\Sc_{\groupindex}$   (see Lines~\ref{line:OciorRMVBAInputSHMDACond} and \ref{line:OciorRMVBAInputSHMDA}   of Algorithm~\ref{algm:OciorRMVBA}).   
Consequently,   $\OciorRBA[(\IDMVBA, \groupindex, \setindex)]$     eventually  outputs the same value $\Me$ at    all hones nodes within  $\Sc_\groupindex$ (See Lines~\ref{line:OciorRMVBAInputRBACond} and \ref{line:OciorRMVBAInputRBA}  of Algorithm~\ref{algm:OciorRMVBA}).  
Therefore, all honest nodes within  $\Sc_{\groupindex}$ eventually set $\DRBCReadyindicator[\setindex] \gets 1, \DRBCFinishindicator[\setindex] \gets 1, \DRBCConfirmindicator[\setindex] \gets 1$. 
In this case, all honest nodes within  $\Sc_{\groupindex}$ eventually go to Line~\ref{line:OciorRMVBALoopCond} and execute the steps in Lines~\ref{line:OciorRMVBALoopCond}-\ref{line:OciorRMVBALoopEnd} of  Algorithm~\ref{algm:OciorRMVBA}.  Based on the result of   Lemma~\ref{lm:RMVBAgoodResOutputValuesALL},  if  $\setindex =0$, then  all honest nodes within  $\Sc_{\groupindex}$ eventually output a consistent value $\Me$. 
For the case where  $\setindex =1$, if  $1\gets \ABBA[\ltuple \IDMVBA, \groupindex, 0, \networksize_\groupindex ,  \networkfaultsize_\groupindex \rtuple]$,  then all honest nodes within  $\Sc_{\groupindex}$ eventually output a consistent value $\Me'$ for some $\Me'$. Otherwise, they eventually output a consistent value $\Me$. 
\end{proof}

\begin{lemma}       \label{lm:RMVBAbv1}
Assume  that $\Sc_\groupindex$ has good resilience.    If $\ABBA[\ltuple \IDMVBA, \groupindex, \electionoutput, \networksize_{\groupindex},  \networkfaultsize_{\groupindex}\rtuple]$ outputs $1$, then at least one honest Node~$i$ within $\Sc_{\groupindex}$  has set $\Vr_i =1$ from $\OciorRBA[(\IDMVBA, \groupindex, \electionoutput)]$, for $\electionoutput\in \{0,1\}$. 
\end{lemma}
\begin{proof}
In this setting, if $\ABBA[\ltuple \IDMVBA, \groupindex, \electionoutput, \networksize_{\groupindex},  \networkfaultsize_{\groupindex}\rtuple]$ outputs $\ABAoutput=1$,  then at least one honest node must have provided an input   $\ABBAoutput=1$ to $\ABBA[\ltuple \IDMVBA, \groupindex, \electionoutput, \networksize_{\groupindex},  \networkfaultsize_{\groupindex}\rtuple]$   (see Line~\ref{line:OciorRMVBAABAoutput} of Algorithm~\ref{algm:OciorRMVBA}), due to the validity property of Byzantine agreement. This means that at least one honest node must have produced an output  $\ABBAoutput=1$ from $\ABBBA[\ltuple \IDMVBA, \groupindex,  \electionoutput, \networksize_{\groupindex},  \networkfaultsize_{\groupindex}\rtuple] (\DRBCReadyindicator[\electionoutput], \DRBCFinishindicator[\electionoutput])$    (see Line~\ref{line:OciorRMVBAABBBAoutput} of Algorithm~\ref{algm:OciorRMVBA}).  
 If an honest node outputs $1$  from $\ABBBA$, then at least one honest node  must have provided at least one input as $1$ to $\ABBBA$  (see Line~\ref{line:ABBBAoneoutput} of Algorithm~\ref{algm:ABBBA}, biased integrity property).              
Thus, if $\ABBA[\ltuple \IDMVBA, \groupindex, \electionoutput, \networksize_{\groupindex},  \networkfaultsize_{\groupindex}\rtuple]$ outputs $\ABAoutput=1$, then at least one honest node  must have set  $\DRBCReadyindicator[\electionoutput] =1$ or  $\DRBCFinishindicator[\electionoutput] =1$. 
If one honest Node~$i$  sets   $\DRBCReadyindicator[\electionoutput] =1$, it must have set $\Vr_i =1$ from $\OciorRBA[(\IDMVBA, \groupindex, \electionoutput)]$  (see Line~\ref{line:RMVBAReadyindicatorCondition}  of Algorithm~\ref{algm:OciorRMVBA}).
If one honest Node~$i$  sets   $\DRBCFinishindicator[\electionoutput] =1$, then at least  $\networksize_{\groupindex}-2\networkfaultsize_{\groupindex}$ honest nodes must have set $\Vr_i =1$ from $\OciorRBA[(\IDMVBA, \groupindex, \electionoutput)]$  (see Line~\ref{line:RMVBAFinishindicatorCondition}  of Algorithm~\ref{algm:OciorRMVBA}). 
Therefore, if $\ABBA[\ltuple \IDMVBA, \groupindex, \electionoutput, \networksize_{\groupindex},  \networkfaultsize_{\groupindex}\rtuple]$ outputs $\ABAoutput=1$, then at least one honest Node~$i$ within $\Sc_{\groupindex}$  has set $\Vr_i =1$ from $\OciorRBA[(\IDMVBA, \groupindex, \electionoutput)]$. 
\end{proof}

\begin{lemma}       \label{lm:RMVBAbv0}
Assume  that $\Sc_\groupindex$ has good resilience.    If $\ABBA[\ltuple \IDMVBA, \groupindex, \electionoutput, \networksize_{\groupindex},  \networkfaultsize_{\groupindex}\rtuple]$ outputs $1$, then no   honest Node~$i$ within $\Sc_{\groupindex}$  will  set $\Vr_i =0$ from $\OciorRBA[(\IDMVBA, \groupindex, \electionoutput)]$, for $\electionoutput\in \{0,1\}$. 
\end{lemma}
\begin{proof}
Lemma~\ref{lm:RMVBAbv1} reveals that if $\ABBA[\ltuple \IDMVBA, \groupindex, \electionoutput, \networksize_{\groupindex},  \networkfaultsize_{\groupindex}\rtuple]$ outputs $\ABAoutput=1$,  then  at least one honest Node~$i$ within $\Sc_{\groupindex}$  has set $\Vr_i =1$ from $\OciorRBA[(\IDMVBA, \groupindex, \electionoutput)]$ in this setting.  If one honest node has set $\Vr_i =1$, then at least $\networksize_{\groupindex}- 2\networkfaultsize_{\groupindex}$   
honest nodes have  sent $\ltuple \SItwo, \IDMVBAtilde, 1 \rtuple$ (see Line~\ref{line:RBAReadyCondition} of Algorithm~\ref{algm:OciorRBA}).  In this case,   the size of  $\Ss_{0}^{[2]}$ is bounded by  $|\Ss_{0}^{[2]}| \leq  \networksize_{\groupindex}- (\networksize_{\groupindex}- 2\networkfaultsize_{\groupindex} )   < \networksize_{\groupindex}- \networkfaultsize_{\groupindex}$ from the view of any honest node, which indicates that  no   honest node   will  set $\Vr_i =0$   from $\OciorRBA[(\IDMVBA, \groupindex, \electionoutput)]$. 
\end{proof}

\begin{lemma}    \label{lm:RMVBAb1ph3one}
Assume  that $\Sc_\groupindex$ has good resilience.    If $\ABBA[\ltuple \IDMVBA, \groupindex, \electionoutput, \networksize_{\groupindex},  \networkfaultsize_{\groupindex}\rtuple]$ outputs $1$, then eventually every honest  node within $\Sc_{\groupindex}$  will  set $\Phthreeindicator= 1$ from $\OciorRBA[(\IDMVBA, \groupindex, \electionoutput)]$, for $\electionoutput\in \{0,1\}$. 
\end{lemma}
\begin{proof}
From the result of Lemma~\ref{lm:RMVBAbv0},    if $\ABBA[\ltuple \IDMVBA, \groupindex, \electionoutput, \networksize_{\groupindex},  \networkfaultsize_{\groupindex}\rtuple]$ outputs $1$,  then  no   honest Node~$i$ within $\Sc_{\groupindex}$  will  set $\Vr_i =0$ from $\OciorRBA[(\IDMVBA, \groupindex, \electionoutput)]$. This suggests that, in this case, no   honest node  will  set $\VrOutput= 0$ and Line~\ref{line:RBAoutputdefault} of Algorithm~\ref{algm:OciorRBA} will never be executed in $\OciorRBA[(\IDMVBA, \groupindex, \electionoutput)]$. 
On the other hand,   if $\ABBA[\ltuple \IDMVBA, \groupindex, \electionoutput, \networksize_{\groupindex},  \networkfaultsize_{\groupindex}\rtuple]$ outputs $\ABAoutput=1$, then   eventually every honest  node within $\Sc_{\groupindex}$  will  
input $\ABAoutput=1$ into $\OciorRBA[(\IDMVBA, \groupindex, \electionoutput)]$ and set $\Phthreeindicator= 1$  (see  Line~\ref{line:RBAph3BB} of Algorithm~\ref{algm:OciorRBA}). 
\end{proof}

\begin{lemma}     \label{lm:RMVBAb1SizeSi1set}
Assume  that $\Sc_\groupindex$ has good resilience.    If $\ABBA[\ltuple \IDMVBA, \groupindex, \electionoutput, \networksize_{\groupindex},  \networkfaultsize_{\groupindex}\rtuple]$ outputs $1$,  then at least $\networksize_{\groupindex}-2\networkfaultsize_{\groupindex} \geq \networkfaultsize_{\groupindex}+1$  honest nodes within $\Sc_{\groupindex}$  have  set  $\Ry_i^{[2]} = 1$   from $\OciorRBA[(\IDMVBA, \groupindex, \electionoutput)]$, for $\electionoutput\in \{0,1\}$ . 
\end{lemma}
\begin{proof}
From the result of Lemma~\ref{lm:RMVBAbv1},      if $\ABBA[\ltuple \IDMVBA, \groupindex, \electionoutput, \networksize_{\groupindex},  \networkfaultsize_{\groupindex}\rtuple]$ outputs $1$, then at least one honest Node~$i$ within $\Sc_{\groupindex}$  has set $\Vr_i =1$ from $\OciorRBA[(\IDMVBA, \groupindex, \electionoutput)]$. When one  honest Node~$i$ within $\Sc_{\groupindex}$  has set $\Vr_i =1$, it means that at least  $\networksize_{\groupindex}-2\networkfaultsize_{\groupindex} \geq \networkfaultsize_{\groupindex}+1$  honest nodes within $\Sc_{\groupindex}$  have  set  $\Ry_i^{[2]} = 1$   from $\OciorRBA[(\IDMVBA, \groupindex, \electionoutput)]$  (see  Line~\ref{line:RBAReadyCondition} of Algorithm~\ref{algm:OciorRBA}).   
\end{proof}

\begin{lemma}   \cite[Lemma 11]{ChenOciorCOOL:24}     \label{lm:RMVBAb1samemessage}
Assume  that $\Sc_\groupindex$ has good resilience.    If $\ABBA[\ltuple \IDMVBA, \groupindex, \electionoutput, \networksize_{\groupindex},  \networkfaultsize_{\groupindex}\rtuple]$ outputs $1$, then all of the honest nodes who set $\Ry_i^{[2]} = 1$  in Phase~$2$ of $\OciorRBA[(\IDMVBA, \groupindex, \electionoutput)]$ should have  the same input message $\Me^{\star}$ at the beginning of Phase~$1$ of $\OciorRBA[(\IDMVBA, \groupindex, \electionoutput)]$, for some $\Me^{\star}$, for $\electionoutput\in \{0,1\}$. 
\end{lemma}
\begin{proof}
 The result is directly derived from  \cite[Lemma 11]{ChenOciorCOOL:24}. 
\end{proof}

\begin{lemma}    [Consistency and Totality Properties of $\OciorRBA$] \label{lm:RMVBAb1ConsistantOutputs}
Assume  that $\Sc_\groupindex$ has good resilience.    If $\ABBA[\ltuple \IDMVBA, \groupindex, \electionoutput, \networksize_{\groupindex},  \networkfaultsize_{\groupindex}\rtuple]$ outputs $1$, then all   honest nodes  within $\Sc_\groupindex$ eventually  output  the same  message $\Me^{\star}$  from $\OciorRBA[(\IDMVBA, \groupindex, \electionoutput)]$, for some $\Me^{\star}$, for $\electionoutput\in \{0,1\}$. 
\end{lemma}
\begin{proof}
 The proof is similar to  \cite[Theorem 5]{ChenOciorCOOL:24}.   
 If $\ABBA[\ltuple \IDMVBA, \groupindex, \electionoutput, \networksize_{\groupindex},  \networkfaultsize_{\groupindex}\rtuple]$ outputs $1$, then we have the  following facts: 
\begin{itemize}
\item    Fact 1:  Eventually every honest  node within $\Sc_{\groupindex}$  will  set $\Phthreeindicator= 1$ and go to Phase~3 of $\OciorRBA[(\IDMVBA, \groupindex, \electionoutput)]$  (Lemma~\ref{lm:RMVBAb1ph3one}).     
\item    Fact 2:  All of the honest nodes who set $\Ry_i^{[2]} = 1$  in Phase~$2$ of $\OciorRBA[(\IDMVBA, \groupindex, \electionoutput)]$ should have  the same input message $\Me^{\star}$ at the beginning of Phase~$1$ of $\OciorRBA[(\IDMVBA, \groupindex, \electionoutput)]$ for some $\Me^{\star}$ (Lemma~\ref{lm:RMVBAb1samemessage}).     
\item    Fact 3:  At least $\networkfaultsize_{\groupindex}+1$  honest nodes within $\Sc_{\groupindex}$  have  set  $\Ry_i^{[2]} = 1$   from $\OciorRBA[(\IDMVBA, \groupindex, \electionoutput)]$ (Lemma~\ref{lm:RMVBAb1SizeSi1set}).     
\end{itemize} 
From Fact 2, if an honest node sets  $\Ry_i^{[2]} = 1$, then this node outputs the value $\Me^{\star}$ in $\OciorRBA[(\IDMVBA, \groupindex, \electionoutput)]$ (see Line~\ref{line:RBAph3SIPhtwoOutput} of Algorithm~\ref{algm:OciorRBA}). 
From Facts 2 and 3,  if an honest Node~$i$ sets  $\Ry_i^{[2]} = 0$, then  it will eventually receives at least $\networkfaultsize_{\groupindex}+1$  matching $\ltuple\SYMBOL, \IDMVBAtilde, (\ECCEnc_i (\networksize_{\groupindex}, \ktilde_{\groupindex}, \Me^{\star}),  *) \rtuple$  messages from the honest nodes within $\Ss_1^{[2]}$,  where $\ECCEnc_i ()$ denotes the $i$th encoded  symbol  and  $\ktilde_{\groupindex}$ is an encoding parameter. 
In this case, Node~$i$ will set $y_i^{(i)} \gets \ECCEnc_i (\networksize_{\groupindex}, \ktilde_{\groupindex}, \Me^{\star})$   in Line~\ref{line:RBAph3MajorityRule}  of Algorithm~\ref{algm:OciorRBA}, and send  $\ltuple \CORRECTSYMBOL, \IDMVBAtilde, y_i^{(i)}  \rtuple$  to  all nodes   in Line~\ref{line:RBAph3SendCorrectSymbols}.  
Therefore,  every symbol $y_j^{(j)}$ sent from honest nodes  and collected in $\OECCorrectSymbolSet$   should be the  symbol  encoded from the same message  $\Me^{\star}$.  Thus,   every honest node who sets   $\Ry_i^{[2]} = 0$ will eventually decode the message $\Me^{\star}$ with OEC decoding and output $\Me^{\star}$ in Line~\ref{line:RBAph3Output2}  of Algorithm~\ref{algm:OciorRBA}.   
\end{proof}

\begin{algorithm}
\caption{$\OciorMVBARelaxedResilienceIT$  protocol, with identifier $\IDMVBA$, for $n\geq 5t+1$. Code is shown for $\Node_{\thisnodeindex}$.}    \label{algm:OciorMVBARelaxedResilienceIT} 
\begin{algorithmic}[1]
\vspace{5pt}    
  
\footnotesize
 
\Statex   \emph{//   **  $\OciorMVBARelaxedResilienceIT$:   without any cryptographic assumption (other than common coin), with a relaxed resilience  $n\geq 5t+1$**}  
 \Statex   \emph{//   ** $\ACD[\IDMVBA]$:  a  protocol for   $n$ parallel asynchronous  complete  information dispersal ($\ACD$)    instances; once an $\ACD$ instance is complete, there exists a retrieval scheme to correctly retrieve its delivered message. **} 
 \Statex   \emph{//   **  $\Election[\ltuple \IDMVBA,  \electionround \rtuple]$: an election protocol,  requiring at least $t+1$ inputs from distinct nodes to generate an output $\electionoutput$, for  $\electionround \in [1:n]$**}   

 \Statex   \emph{//   **   $\AMBARelaxedResilienceIT[\ltuple \IDMVBA,  \electionoutput \rtuple]$  calls  the  asynchronous Byzantine agreement ($\ABA$) protocol  by Li-Chen \cite{LCabaISIT:21}, for $n\geq 5t+1$, using only $O(1)$ common coins,  $O(n |\MVBAInputMsg|+  n^2\log \alphabetsize)$ total bits, and  $O(1)$   rounds, without any cryptographic assumption (other than common coin) **}

 \Statex
 
\State {\bf upon} receiving MVBA input  message  $\MVBAInputMsg_{\thisnodeindex}$  and $\Predicate(\MVBAInputMsg_{\thisnodeindex}) \eqlog \true$ {\bf do}:  
\Indent  
	\State  $\ShareRecord \gets \ACDRRIT[ \IDMVBA  ](\MVBAInputMsg_{\thisnodeindex})$           \quad\quad   \quad  \quad\quad   \quad   \  \emph{//    a  protocol for   $n$ parallel $\ACD$    instances }
	\For {$\electionround \inset [1:n]$}                        \label{line:MVBARRround}	
	
		\State  $\electionoutput \gets \Election[\ltuple \IDMVBA,  \electionround \rtuple]$        \label{line:MVBARRelection}     \   	\quad \quad \quad   \quad \quad \quad  \quad   \emph{//      an election protocol } 
		\State  $\bar{\wv}   \gets \DataRetrievalRelaxedResilienceIT[\ltuple \IDMVBA,  \electionoutput \rtuple](\ShareRecord[\electionoutput])$              \quad     \quad      \quad\quad  \emph{//      shuffle the code symbols originally sent from Node~$\electionoutput$ and decode } 	
		\State  $\hat{\MVBAInputMsg}  \gets \AMBARelaxedResilienceIT[\ltuple \IDMVBA,  \electionoutput \rtuple](\bar{\wv})$          \label{line:MVBARRaba}  	  \quad\quad  \quad\quad   \quad\quad \emph{//       call the asynchronous BA   protocol  by Li-Chen \cite{LCabaISIT:21} }

				\If {$\Predicate(\hat{\MVBAInputMsg}) \eqlog \true$}            \label{line:MVBARRpred}  
					\State  $\Output$  $\hat{\MVBAInputMsg}$ and $\terminate$ 					  				
				\EndIf

	\EndFor

\EndIndent

\end{algorithmic}
\end{algorithm}

\begin{algorithm}
\caption{$\ACDRRIT$ subprotocol with identifier $\IDMVBA$ for  $t<\frac{n}{5}$.  Code is shown for $\Node_{\thisnodeindex}$ }    \label{algm:ACDRRIT} 
\begin{algorithmic}[1]
\vspace{5pt}    
\footnotesize

\Statex   \emph{//   **  $\ACD[  \IDMVBA  ]$ is a protocol for $n$ parallel $\ACD$ instances   $\ACD[\ltuple \IDMVBA,  1 \rtuple],  \ACD[\ltuple \IDMVBA,  2 \rtuple], \cdots,  \ACD[\ltuple \IDMVBA,  n \rtuple]$  **} 
\Statex   \emph{//   **    $\ACD[\ltuple \IDMVBA, j \rtuple]$ is an $\ACD$ instance for delivering the message proposed from Node~$j$ ** } 		
\Statex   \emph{//   ** Once Node~$j$  completes $\ACD[\ltuple \IDMVBA, j \rtuple]$, there exists a retrieval scheme to correctly retrieve the message. **} 
\Statex   \emph{//   ** When an honest node returns and stops this protocol, then at least $n-t$  $\ACD$ instances  have been completed. **} 

\Statex

 \State  Initially set  $\ShareRecord[j]\gets    \defaultvalue $,   $\forall j\in [1:n]$
 
\Statex
 
\Statex   \emph{//   **  $\ACD$-share  **} 	 	
 
\State {\bf upon} receiving input  message  $\MVBAInputMsg_{\thisnodeindex}$ {\bf do}:  
\Indent  
	\State $[\EncodedSymbol_{1}, \EncodedSymbol_{2}, \cdots, \EncodedSymbol_{n} ]   \gets \ECCEnc(n, t+1, \MVBAInputMsg_{\thisnodeindex})$    
\State   $\send$  $(\SHARE,  \IDMVBA,    \EncodedSymbol_{j})$ to $\Node_j$,  $\forall j \in  [1:\networksizen]$  \  \quad\quad \quad \quad\quad\quad \quad \quad \emph{// exchange coded  symbols }   		
\EndIndent

\Statex

\Statex   \emph{//   **  $\ACD$-vote  **} 	 

\State {\bf upon} receiving   $(\SHARE,  \IDMVBA,  \EncodedSymbol)$ from  $\Node_j$ for the first time {\bf do}:   
\Indent  
		\State  $\ShareRecord[j] \gets    \EncodedSymbol $  
		\State $\send$ $(\VOTE, \IDMVBA)$ to   $\Node_j$   
\EndIndent

 \Statex

\Statex   \emph{//   **  vote for election **}

\State {\bf upon} receiving   $n-t$  $(\VOTE, \IDMVBA)$  messages from distinct nodes {\bf do}:  
\Indent  
	\State $\send$ $(\ELECTION, \IDMVBA)$ to  all nodes     \quad \quad   \emph{//    $\ACD[\ltuple \IDMVBA,  \thisnodeindex \rtuple]$ is complete at this point } 	
 
\EndIndent

\Statex

\Statex   \emph{//   ** confirm for election **}

\State {\bf upon} receiving   $n-t$  $(\ELECTION, \IDMVBA)$  messages from distinct nodes and  $(\CONFIRM, \IDMVBA)$  not yet sent {\bf do}:  
\Indent  
	\State $\send$ $(\CONFIRM, \IDMVBA)$ to  all nodes       	
 
\EndIndent

\State {\bf upon} receiving   $t+1$  $(\CONFIRM, \IDMVBA)$  messages from distinct nodes and  $(\CONFIRM, \IDMVBA)$  not yet sent {\bf do}:  
\Indent  
	\State $\send$ $(\CONFIRM, \IDMVBA)$ to  all nodes       	
 
\EndIndent

\Statex

\Statex   \emph{//   ** return and stop **} 
\State {\bf upon} receiving   $2t+1$  $(\CONFIRM, \IDMVBA)$  messages from distinct nodes {\bf do}:  
\Indent  
	\If {$(\CONFIRM, \IDMVBA)$   not yet sent }     
		\State $\send$ $(\CONFIRM, \IDMVBA)$ to  all nodes       	
	\EndIf
	\State  $\Return$  $\ShareRecord$	  
\EndIndent

\end{algorithmic}
\end{algorithm}

\begin{algorithm}
\caption{$\DataRetrievalRelaxedResilienceIT$  subprotocol for  $t<\frac{n}{5}$, with identifier $\IDMVBAtilde=\ltuple \IDMVBA,  \electionoutput \rtuple$. Code is shown for $\Node_{\thisnodeindex}$.}    \label{algm:DataRetrievalRelaxedResilienceIT} 
\begin{algorithmic}[1]
\vspace{5pt}    
 
\footnotesize

\State Initially set $\CodedSymbols [\electionoutput] \gets \{\}$  		

\State {\bf upon} receiving  input  $\ShareRecord[\electionoutput]$,  for $\ShareRecord[\electionoutput]:=  \EncodedSymbol^{\star}$   {\bf do}:  		 
\Indent

					\State $\send$ $(\ECHOSHARE, \IDMVBA, \electionoutput,    \EncodedSymbol^{\star} )$  to all nodes   
 
				\State {\bf wait} for  $|\CodedSymbols[\electionoutput] |=n-t$   
				\Indent
					\State  $\hat{\MVBAInputMsg}  \gets\ECCDec(n, t+1, \CodedSymbols[\electionoutput])$       
					\State  $\Return$  $\hat{\MVBAInputMsg}$					   
				\EndIndent

\EndIndent

\State {\bf upon} receiving   $(\ECHOSHARE, \IDMVBA, \electionoutput,    \EncodedSymbol)$ from  $\Node_j$ for the first time  {\bf do}:    
\Indent  
	\State  $\CodedSymbols[\electionoutput] \gets   \CodedSymbols[\electionoutput] \cup \{j: \EncodedSymbol\}$
\EndIndent

\end{algorithmic}
\end{algorithm}

\section{$\OciorMVBARelaxedResilienceIT$}\label{sec:OciorMVBARelaxedResilienceIT}

This proposed $\OciorMVBARelaxedResilienceIT$ is an error-free,  information-theoretically secure asynchronous $\MVBA$  protocol, with relaxed resilience  $n\geq 5t+1$.  
$\OciorMVBARelaxedResilienceIT$ does not rely on any   cryptographic assumptions, such as  signatures or hashing,  except for  the common coin assumption.

\begin{figure}  
\centering
\includegraphics[width=17.5cm]{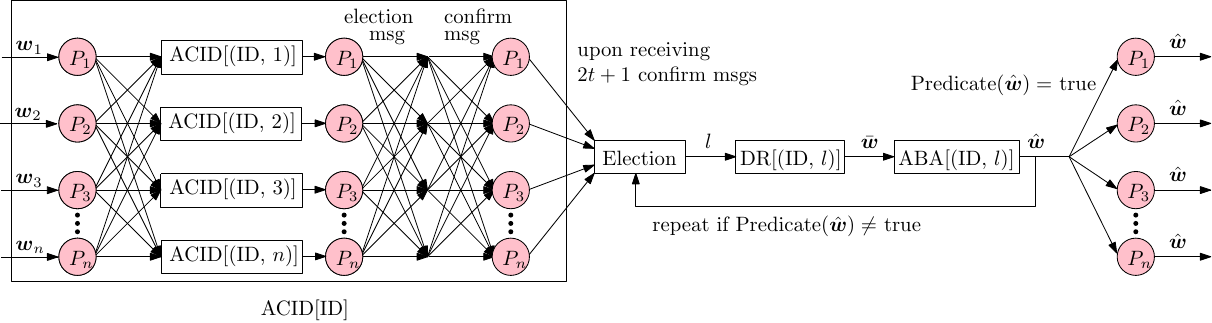}
\caption{A block diagram of the proposed $\OciorMVBARelaxedResilienceIT$ protocol with an identifier $\IDMVBA$. 
}
\label{fig:OciorMVBARelaxedResilienceIT}
\end{figure}

\subsection{Overview of the  proposed $\OciorMVBARelaxedResilienceIT$ protocol}

The proposed  $\OciorMVBARelaxedResilienceIT$ is described in Algorithm~\ref{algm:OciorMVBARelaxedResilienceIT}, along with   Algorithms~\ref{algm:ACDRRIT} and  \ref{algm:DataRetrievalRelaxedResilienceIT}. Fig.~\ref{fig:OciorMVBARelaxedResilienceIT} presents a block diagram of the proposed $\OciorMVBARelaxedResilienceIT$ protocol. 
The proposed  $\OciorMVBARelaxedResilienceIT$ consists of the algorithms $\ACDRRIT[ \IDMVBA  ]$, $\Election[\ltuple \IDMVBA,  \electionround \rtuple]$, $\DataRetrievalRelaxedResilienceIT[\ltuple \IDMVBA,  \electionoutput \rtuple]$,  and $\AMBARelaxedResilienceIT[\ltuple \IDMVBA,  \electionoutput \rtuple]$ for   $\electionround, \electionoutput\in [1:n]$.
\begin{itemize}
\item    $\ACDRRIT[ \IDMVBA  ]$:  This  is a  protocol for   $n$ parallel  $\ACD$    instances: $\ACD[\ltuple \IDMVBA,  1 \rtuple]$,  $\ACD[\ltuple \IDMVBA,  2 \rtuple], \cdots$,  $\ACD[\ltuple \IDMVBA,  n \rtuple]$.  Once an $\ACD$ instance is complete, there exists a retrieval scheme to correctly retrieve its delivered message.
\item  $\Election[\ltuple \IDMVBA,  \electionround \rtuple]$: This is an election protocol that  requires  at least $t+1$ inputs from distinct nodes to generate a random value $\electionoutput$, where  $\electionround \in [1:n]$.
\item  $\DataRetrievalRelaxedResilienceIT[\ltuple \IDMVBA,  \electionoutput \rtuple]$: An $\ACD[\ltuple \IDMVBA, \electionoutput \rtuple]$ protocol is complemented by a data retrieval protocol $\DataRetrieval[\ltuple \IDMVBA, \electionoutput \rtuple]$,  in which each node retrieves the   message proposed by $\Node_\electionoutput$ from $n$ distributed nodes.  
\begin{itemize}
\item If  Node~$\electionoutput$ is honest and has completed $\ACD[\ltuple \IDMVBA, \electionoutput\rtuple]$, then during $ \DataRetrievalRelaxedResilienceIT[\ltuple \IDMVBA,  \electionoutput \rtuple]$, each node will receive at least $2t+1$ shares	  generated from Node~$\electionoutput$, given  $n\geq 5t+1$. In this case all honest node eventually output the same message from $ \DataRetrievalRelaxedResilienceIT[\ltuple \IDMVBA,  \electionoutput \rtuple]$. 
\item Even if  Node~$\electionoutput$ is dishonest,     $\AMBARelaxedResilienceIT[\ltuple \IDMVBA,  \electionoutput \rtuple]$  ensures that all honest nodes output the same message. 
\end{itemize} 
\item $\AMBARelaxedResilienceIT[\ltuple \IDMVBA,  \electionoutput \rtuple]$ :  This is an  asynchronous Byzantine agreement  protocol that calls the protocol  by Li-Chen \cite{LCabaISIT:21} for $n\geq 5t+1$. It  uses only $O(1)$ common coins,  $O(n |\MVBAInputMsg|+  n^2\log \alphabetsize)$ total bits, and  $O(1)$   rounds, without relying on any cryptographic assumptions,   except for  the common coin assumption. 
\end{itemize}

\subsection{Analysis of $\OciorMVBARelaxedResilienceIT$}    \label{sec:AnalysisOciorMVBARelaxedResilienceIT}

\begin{theorem}  [Agreement]  \label{thm:MVBARRagreement}
In $\OciorMVBARelaxedResilienceIT$,  given $n\geq 5t+1$, if any two honest nodes output $\wv'$ and $\wv''$, respectively, then  $\wv'=\wv''$.        
\end{theorem}
\begin{proof}
In $\OciorMVBARelaxedResilienceIT$,  if any two honest nodes output values at Rounds  $\electionround$ and $\electionround'$ (see Line~\ref{line:MVBARRround} of Algorithm~\ref{algm:OciorMVBARelaxedResilienceIT}), respectively, then $\electionround=\electionround'$, due to the consistency property of the protocols  $\Election[\ltuple \IDMVBA,  \electionround \rtuple]$  and $\AMBARelaxedResilienceIT[\ltuple \IDMVBA,  \electionoutput \rtuple]$. 
Moreover,  at  the same round $\electionround$,  if any two honest nodes output $\wv'$ and $\wv''$, respectively, then  $\wv'=\wv''$, due to the consistency property of the protocol   $\AMBARelaxedResilienceIT[\ltuple \IDMVBA,  \electionoutput \rtuple]$. 
\end{proof}

\begin{theorem}  [Termination]  \label{thm:MVBARRTermination}
In $\OciorMVBARelaxedResilienceIT$,  given $n\geq 5t+1$, every honest node eventually outputs a value and terminates.        
\end{theorem}
\begin{proof}
In this setting, every honest node eventually returns $\ShareRecord$ and terminates from the protocol $\ACDRRIT[ \IDMVBA ]$, due to the Termination property of this protocol.
Furthermore, by the  Integrity property of  $\ACDRRIT[ \IDMVBA ]$, if an honest node  returns $\ShareRecord$ and terminates from the protocol $\ACDRRIT[ \IDMVBA ]$, then there exists a set $\Ic^{\star}$ such that the following conditions hold: ) $\Ic^{\star}\subseteq [1:n]\setminus \Fc$, where $\Fc$ denotes the set of indexes of all dishonest nodes; 2)  $|\Ic^{\star}| \geq n-2t$; and  3) for any $i\in \Ic^{\star}$,  $\Node_i$ has completed the dispersal $\ACD[\ltuple \IDMVBA, i \rtuple]$.      

Subsequently, every honest node eventually runs  $\electionoutput \gets \Election[\ltuple \IDMVBA,  \electionround \rtuple]$  in Line~\ref{line:MVBARRelection} of Algorithm~\ref{algm:OciorMVBARelaxedResilienceIT}, at the same round $\electionround$.
If  Node~$\electionoutput$ is honest and  $\electionoutput\in  \Ic^{\star}$, then during $ \DataRetrievalRelaxedResilienceIT[\ltuple \IDMVBA,  \electionoutput \rtuple]$, each node will receive at least $2t+1$ shares	  generated from Node~$\electionoutput$, given  $n\geq 5t+1$. In this case all honest node eventually output the same message  from both $ \DataRetrievalRelaxedResilienceIT[\ltuple \IDMVBA,  \electionoutput \rtuple]$ and $\AMBARelaxedResilienceIT[\ltuple \IDMVBA,  \electionoutput \rtuple]$, and then terminate. 
If  Node~$\electionoutput$ is dishonest and the message $\hat{\MVBAInputMsg}$ output by  $\AMBARelaxedResilienceIT[\ltuple \IDMVBA,  \electionoutput \rtuple]$ in Line~\ref{line:MVBARRaba} does not satisfy   $\Predicate(\hat{\MVBAInputMsg}) = \true$, then all honest nodes proceed to the next round. 
All honest node eventually terminates if, at some round~$\electionround$, $\Election[\ltuple \IDMVBA,  \electionround \rtuple]$  outputs a value $\electionoutput$ such that  Node~$\electionoutput$ is honest and $\electionoutput\in  \Ic^{\star}$. 	
\end{proof}

\begin{theorem}  [External Validity]  \label{thm:MVBARRExternalvalidity}
In $\OciorMVBARelaxedResilienceIT$,  if an honest node outputs a value $\MVBAInputMsg$, then $\Predicate(\MVBAInputMsg)=\true$.    
\end{theorem}
\begin{proof}
 In $\OciorMVBARelaxedResilienceIT$,  if an honest node outputs a value $\MVBAInputMsg$,  it has verified that $\Predicate(\MVBAInputMsg)=\true$ at Line~\ref{line:MVBARRpred} of Algorithm~\ref{algm:OciorMVBARelaxedResilienceIT}.
\end{proof}

\begin{algorithm} [H]
\caption{$\OciorMVBAHash$  protocol, with identifier $\IDMVBA$, for $n\geq 3t+1$. Code is shown for $\Node_{\thisnodeindex}$.}    \label{algm:OciorMVBAh} 
\begin{algorithmic}[1]
\vspace{5pt}    
\footnotesize

\Statex   \emph{//   **  Merkle tree is implemented here for vector commitment  based on hashing    **} 	
\Statex   \emph{//   **  $\VCCom( )$     outputs a  commitment, i.e., Merkle root, with $O(\kappa)$ bits **} 	
\Statex   \emph{//   **  $\VCOpen()$ returns a proof that the targeted value  is the committed   element of  the vector   **} 	
\Statex   \emph{//   **  $\VCVerify(j, \vectorcommitment,  \EncodedSymbol,   \proofpositionvc)$ returns true only if  $\proofpositionvc$ is a valid  proof that $\vectorcommitment$  is the commitment of  a vector whose $j$th element is   $ \EncodedSymbol$   **}

\Statex

\Statex

\State {\bf upon} receiving MVBA input  message  $\MVBAInputMsg_{\thisnodeindex}$  and $\Predicate(\MVBAInputMsg_{\thisnodeindex}) \eqlog \true$ {\bf do}:  
\Indent  
	\State  $[\LockRecord, \ReadyRecord,  \FinishRecord, \ShareRecord] \gets \ACDh[ \IDMVBA  ](\MVBAInputMsg_{\thisnodeindex})$            \quad   \quad   \  \emph{//    a  protocol for   $n$ parallel $\ACD$    instances  }
	\For {$\electionround \inset [1:n]$}   \label{line:MVBAHashABAround} 
	
		\State  $\electionoutput \gets \Election[\ltuple \IDMVBA,  \electionround \rtuple]$    \label{line:MVBAHashElection}   \quad \quad \quad \quad \quad \quad 	\quad \quad \quad   \quad \quad \quad  \quad   \emph{//      an election protocol } 
		\State  $\ABBAoutput\gets \ABBBA[\ltuple \IDMVBA,   \electionoutput\rtuple](\ReadyRecord[\electionoutput], \FinishRecord[\electionoutput])$      \quad\quad \quad\quad  \emph{//      asynchronous  biased binary Byzantine agreement   ($\ABBBA$)   } 
		\State  $\ABAoutput\gets \ABBA[\ltuple \IDMVBA,  \electionoutput\rtuple](\ABBAoutput)$            \  \quad\quad\quad \quad \quad  \quad\quad \quad \quad \quad  \quad\quad \emph{//     an asynchronous  binary Byzantine agreement  ($\ABBA$)   }

		\If {$\ABAoutput \eqlog 1$}     \label{line:MVBAHashABAoutput1}

                            \State $\hat{\MVBAInputMsg}_\electionoutput\gets \DataRetrievalh[\ltuple \IDMVBA,  \electionoutput \rtuple] (\LockRecord[\electionoutput], \ShareRecord[\electionoutput])$	 \     \quad \quad \quad  \quad  \emph{//   Data Retrieval (DR)    } 
				\If {$\Predicate(\hat{\MVBAInputMsg}_\electionoutput) \eqlog \true$}      \label{line:MVBAHashPredicateCond}
					\State  $\Output$  $\hat{\MVBAInputMsg}_\electionoutput$ and $\terminate$ 					  				
				\EndIf

		\EndIf

	\EndFor

\EndIndent

\end{algorithmic}
\end{algorithm}

\begin{algorithm}  [H]
\caption{$\ACDh$ subprotocol with identifier $\IDMVBA$, based on hashing.  Code is shown for $\Node_{\thisnodeindex}$}    \label{algm:ACDh} 
\begin{algorithmic}[1]
\vspace{5pt}       
\footnotesize

\Statex   \emph{//   **  $\ACD[  \IDMVBA  ]$ is a protocol for $n$ parallel $\ACD$ instances   $\ACD[\ltuple \IDMVBA,  1 \rtuple],  \ACD[\ltuple \IDMVBA,  2 \rtuple], \cdots,  \ACD[\ltuple \IDMVBA,  n \rtuple]$ **} 
\Statex   \emph{//   **    $\ACD[\ltuple \IDMVBA, j \rtuple]$ is an $\ACD$ instance for delivering the message proposed from Node~$j$ ** } 		
\Statex   \emph{//   ** Once Node~$j$  completes $\ACD[\ltuple \IDMVBA, j \rtuple]$, there exists a retrieval scheme to correctly retrieve the message **} 
\Statex   \emph{//   ** When an honest node returns and stops this protocol, then at least $n-t$  $\ACD$ instances  have been completed **} 
 \Statex   \emph{//   **  $\ECEnc()$ and $\ECDec()$ are   encoding function and decoding function  of $(n, k)$ erasure code **}

\Statex

\State \emph{// initialization:}
\Indent  

\State $\LockRecord \gets \{\}; \ReadyRecord \gets \{\};  \FinishRecord\gets \{\}; \ShareRecord\gets\{\}; \HashRecord\gets\{\}$  

	\For {$j \inset [1:n]$}	 
		\State $\ShareRecord[j]\gets   \ltuple\defaultvalue,  \defaultvalue,  \defaultvalue \rtuple  ;  \LockRecord[j]\gets 0;\ReadyRecord[j]\gets 0;  \FinishRecord[j]\gets 0 $  	
	\EndFor

\EndIndent

\Statex

\Statex   \emph{//   **  $\ACD$-share  **}

\State {\bf upon} receiving input  message  $\MVBAInputMsg_{\thisnodeindex}$ {\bf do}:  
\Indent  

	\State $[\EncodedSymbol_{1}, \EncodedSymbol_{2}, \cdots, \EncodedSymbol_{n} ]   \gets \ECEnc(n, t+1, \MVBAInputMsg_{\thisnodeindex})$       
	\State $\vectorcommitment \gets \VCCom([\EncodedSymbol_{1}, \EncodedSymbol_{2}, \cdots, \EncodedSymbol_{n} ] )$ 
	\For {$j \inset [1:n]$}
		\State $\proofpositionvc_{j} \gets \VCOpen(\vectorcommitment, \EncodedSymbol_{j}, j)$  
		\State $\send$ $(\SHARE,  \IDMVBA, \vectorcommitment,  \EncodedSymbol_{j}, \proofpositionvc_{j})$ to  $\Node_j$
	\EndFor

\EndIndent

\Statex

\Statex   \emph{//   **  $\ACD$-vote  **} 	 

\State {\bf upon} receiving   $(\SHARE,  \IDMVBA,\vectorcommitment, \EncodedSymbol, \proofpositionvc)$ from  $\Node_j$ for the first time {\bf do}:   
\Indent  
	\If {$\VCVerify(\thisnodeindex, \vectorcommitment,  \EncodedSymbol,   \proofpositionvc) \eqlog \true$}  
		\State  $\ShareRecord[j] \gets \ltuple \vectorcommitment,  \EncodedSymbol,   \proofpositionvc \rtuple; \HashRecord[\vectorcommitment] \gets j$   
		\State $\send$ $(\VOTE, \IDMVBA,   \vectorcommitment)$ to  all nodes   
	\EndIf    
\EndIndent

\Statex

\Statex   \emph{//   **  $\ACD$-lock **} 

\State {\bf upon} receiving  $n-t$  $(\VOTE, \IDMVBA, \vectorcommitment)$  messages from distinct nodes,  for   the same  $\vectorcommitment$ {\bf do}:   \label{line:ACDhreadyCond} 
\Indent  
		\State  $\wait$ until  $\vectorcommitment \in  \HashRecord$ 
		\State $\jstar\gets\HashRecord[\vectorcommitment]; \LockRecord[\jstar]\gets 1$     \label{line:ACDhready} 
		\State $\send$ $(\LOCK, \IDMVBA, \vectorcommitment)$ to  all nodes    
\EndIndent

\Statex 

\Statex   \emph{//   **  $\ACD$-ready **}  

\State {\bf upon} receiving  $n-t$  $(\LOCK, \IDMVBA,  \vectorcommitment)$  messages from distinct nodes, for   the same  $\vectorcommitment$ {\bf do}:  
\Indent  
	\State  $\wait$ until  $\vectorcommitment \in  \HashRecord$ 
	\State $\jstar\gets\HashRecord[\vectorcommitment]; \ReadyRecord[\jstar] \gets 1$       
	\State $\send$ $(\READY, \IDMVBA, \vectorcommitment)$  to  all nodes         
\EndIndent

\Statex

\Statex   \emph{//   **  $\ACD$-finish **}  

\State {\bf upon} receiving  $n-t$  $(\READY, \IDMVBA,  \vectorcommitment)$  messages from distinct nodes,  for   the same  $\vectorcommitment$ {\bf do}:  
\Indent  
	\State  $\wait$ until  $\vectorcommitment \in  \HashRecord$ 
	\State $\jstar\gets\HashRecord[\vectorcommitment]; \FinishRecord[\jstar] \gets 1$    

	\State $\send$ $(\FINISH, \IDMVBA)$  to $\Node_{\jstar}$      
\EndIndent

\Statex

\Statex   \emph{//   **  vote for election **}

\State {\bf upon} receiving   $n-t$  $(\FINISH, \IDMVBA)$  messages from distinct nodes {\bf do}:  
\Indent  
	\State $\send$ $(\ELECTION, \IDMVBA)$ to  all nodes     \quad \quad   \emph{//      $\ACD[\ltuple \IDMVBA,  \thisnodeindex \rtuple]$ is complete at this point   } 	
 
\EndIndent

\Statex

\Statex   \emph{//   ** confirm for election **}

\State {\bf upon} receiving   $n-t$  $(\ELECTION, \IDMVBA)$  messages from distinct nodes and  $(\CONFIRM, \IDMVBA)$  not yet sent {\bf do}:  
\Indent  
	\State $\send$ $(\CONFIRM, \IDMVBA)$ to  all nodes       	
 
\EndIndent

\State {\bf upon} receiving   $t+1$  $(\CONFIRM, \IDMVBA)$  messages from distinct nodes and  $(\CONFIRM, \IDMVBA)$  not yet sent {\bf do}:  
\Indent  
	\State $\send$ $(\CONFIRM, \IDMVBA)$ to  all nodes       	
 
\EndIndent

\Statex

\Statex   \emph{//   ** return and stop **} 
\State {\bf upon} receiving   $2t+1$  $(\CONFIRM, \IDMVBA)$  messages from distinct nodes {\bf do}:  
\Indent  
	\If {$(\CONFIRM, \IDMVBA)$   not yet sent }     
		\State $\send$ $(\CONFIRM, \IDMVBA)$ to  all nodes       	
	\EndIf
	\State  $\Return$  $[\LockRecord, \ReadyRecord,  \FinishRecord, \ShareRecord]$	  
\EndIndent

\end{algorithmic}
\end{algorithm}

\begin{algorithm}  [H]
\caption{$\DataRetrievalh$  subprotocol, with identifier $\IDMVBAtilde=\ltuple \IDMVBA,  \electionoutput \rtuple$, based on hashing. Code is shown for $\Node_{\thisnodeindex}$.}    \label{algm:DataRetrievalh} 
\begin{algorithmic}[1]
\vspace{5pt}    
\footnotesize

\State Initially set  $\CodedSymbols[\electionoutput] \gets \{\}$ 		

\State {\bf upon} receiving  input  $(\lockindicator, \share)$   {\bf do}:  		
\Indent

				\If {$(\lockindicator \eqlog 1)\AND(\share \neq   \ltuple\defaultvalue,  \defaultvalue,  \defaultvalue \rtuple )$}   
					\State $\ltuple \vectorcommitment^{\star},  \EncodedSymbol^{\star},   \proofpositionvc^{\star} \rtuple \gets \share$  					
					\State $\send$ $(\ECHOSHARE, \IDMVBA, \electionoutput, \vectorcommitment^{\star},  \EncodedSymbol^{\star}, \proofpositionvc^{\star})$  to all nodes
				\EndIf

				\State {\bf wait} for   $|\CodedSymbols[\electionoutput][\vectorcommitment] |=t+1$  for some  $\vectorcommitment$
				\Indent
					\State  $\hat{\MVBAInputMsg}  \gets\ECDec(n, t+1, \CodedSymbols[\electionoutput][\vectorcommitment])$  
					\If {$\VCCom( \ECEnc(n, t+1, \hat{\MVBAInputMsg})) = \vectorcommitment$}					
						\State  $\Return$  $\hat{\MVBAInputMsg}$					  
					\Else
						\State  $\Return$  $\bot$					  			
					\EndIf
				\EndIndent

\EndIndent

\State {\bf upon} receiving   $(\ECHOSHARE, \IDMVBA, \electionoutput, \vectorcommitment,  \EncodedSymbol, \proofpositionvc)$ from  $\Node_j$ for the first time, for some $\vectorcommitment, \EncodedSymbol, \proofpositionvc$ {\bf do}:  
\Indent  
	\If {$\VCVerify(j, \vectorcommitment,  \EncodedSymbol,   \proofpositionvc) \eqlog \true$}  
		\If {$\vectorcommitment\notin \CodedSymbols[\electionoutput] $}  
			\State  $\CodedSymbols[\electionoutput][\vectorcommitment]\gets \{j: \EncodedSymbol\}$ 
		\Else
			\State  $\CodedSymbols[\electionoutput][ \vectorcommitment] \gets \CodedSymbols[\electionoutput][ \vectorcommitment] \cup \{j: \EncodedSymbol\}$ 
		\EndIf   
	\EndIf    
\EndIndent

\end{algorithmic}
\end{algorithm}

\begin{figure}  [H]
\centering
\includegraphics[width=16cm]{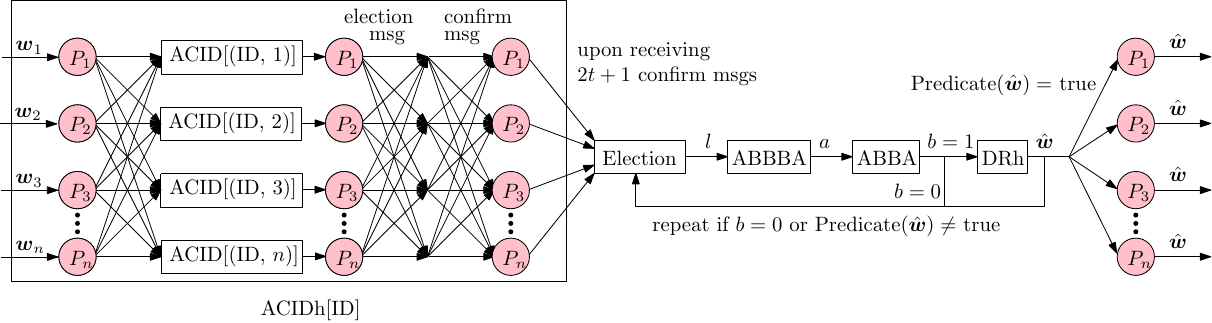}
\caption{A block diagram of the proposed $\OciorMVBAHash$ protocol with an identifier $\IDMVBA$.
}
\label{fig:OciorMVBAHash}
\end{figure}

\section{$\OciorMVBAHash$}\label{sec:OciorMVBAHash}

This proposed $\OciorMVBAHash$ is a hash-based asynchronous $\MVBA$ protocol.
$\OciorMVBAHash$ achieves consensus with a communication complexity of $O(n |\MVBAInputMsg| + n^3)$ bits, an expected round complexity of $O(1)$ rounds, and an expected $O(1)$ number of common coins, given $n \geq 3t + 1$.

\subsection{Overview of the  proposed $\OciorMVBAHash$ protocol}

The proposed  $\OciorMVBAHash$ is described in Algorithm~\ref{algm:OciorMVBAh}, along with   Algorithms~\ref{algm:ABBBA}, \ref{algm:ACDh}, and \ref{algm:DataRetrievalh}. Fig.~\ref{fig:OciorMVBAHash} presents a block diagram of the proposed $\OciorMVBAHash$ protocol.  
In this protocol, we use a vector commitment implemented with a Merkle tree based on hashing. \\

\noindent  {\bf Vector commitment.}   A vector commitment consists of the following algorithms: 
\begin{itemize}
\item    $\VCCom(\yv) \to \vectorcommitment$:  Given an input vector $\yv= [\EncodedSymbol_{1}, \EncodedSymbol_{2}, \cdots, \EncodedSymbol_{n} ]$ of size $n$,  this algorithm  outputs a  commitment $\vectorcommitment$, i.e., Merkle root, with $O(\kappa)$ bits. 
\item   $\VCOpen(\vectorcommitment, \EncodedSymbol_{j}, j) \to \proofpositionvc_{j}$: Given inputs $(\vectorcommitment, \EncodedSymbol_{j}, j)$,  this algorithm returns a value $\proofpositionvc_{j}$ to prove that the targeted value $\EncodedSymbol_{j}$ is the   $j$th committed element of  the vector. 
\item $\VCVerify(j, \vectorcommitment,  \EncodedSymbol_{j},   \proofpositionvc_{j}) \to \true/\false$:   This algorithm  returns true only if  $\proofpositionvc_{j}$ is a valid  proof that $\vectorcommitment$  is the commitment of  a vector whose $j$th element is   $\EncodedSymbol_{j}$. \\
\end{itemize}

The proposed  $\OciorMVBAHash$ consists of the algorithms $\ACDh[ \IDMVBA  ]$, $\Election[\ltuple \IDMVBA,  \electionround \rtuple]$, $\ABBBA[\ltuple \IDMVBA,   \electionoutput \rtuple]$, $\ABBA[\ltuple \IDMVBA,   \electionoutput\rtuple]$, and  $\DataRetrievalh[\ltuple \IDMVBA,  \electionoutput \rtuple]$,  for   $\electionround, \electionoutput\in [1:n]$.
\begin{itemize}
\item    $\ACDRRIT[ \IDMVBA  ]$:  This  is a  protocol for   $n$ parallel  $\ACD$    instances: $\ACD[\ltuple \IDMVBA,  1 \rtuple]$,  $\ACD[\ltuple \IDMVBA,  2 \rtuple], \cdots$,  $\ACD[\ltuple \IDMVBA,  n \rtuple]$.  Once an $\ACD$ instance is complete, there exists a retrieval scheme to correctly retrieve its delivered message.
\item  $\Election[\ltuple \IDMVBA,  \electionround \rtuple]$: This is an election protocol that  requires  at least $t+1$ inputs from distinct nodes to generate a random value $\electionoutput$, where  $\electionround \in [1:n]$.
\item  $\ABBBA[\ltuple \IDMVBA,   \electionoutput \rtuple]$:   This is an asynchronous  biased  binary BA protocol. It has two inputs $(\abbainputA, \abbainputB)$, for some $\abbainputA, \abbainputB \in \{0,1\}$. It guarantees the following properties: 
 \begin{itemize}
\item   {Conditional termination:} Under an input condition---i.e.,  if one honest node inputs its second number as $\abbainputB =1$ then at least $t+1$ honest nodes  input  their first numbers as $\abbainputA =1$---then every honest node eventually outputs a value and terminates.     
\item   { Biased validity:} If $t+1$ honest nodes input the second number as $\abbainputB=1$, then any honest node that terminates outputs $1$.       
\item   { Biased integrity:} If an honest node outputs $1$, then at least one honest node inputs $\abbainputA=1$ or $\abbainputB=1$.             
\end{itemize}
\item  $\ABBA[\ltuple \IDMVBA,   \electionoutput \rtuple]$:   This is an asynchronous   binary $\BA$   protocol.  
\item  $\DataRetrievalh[\ltuple \IDMVBA,  \electionoutput \rtuple]$: This is a data retrieval protocol associated with an $\ACD[\ltuple \IDMVBA, \electionoutput \rtuple]$ protocol.  It is activated only if  $\ABAoutput=1$ (see Line~\ref{line:MVBAHashABAoutput1}  of Algorithm~\ref{algm:OciorMVBAh}), where    $\ABAoutput$ is the output of $\ABBA[\ltuple \IDMVBA,  \electionoutput\rtuple]$.  
 \begin{itemize}
\item   The instance of $\ABAoutput=1$ reveals that at least one honest node outputs $\ABBAoutput=1$ from  $\ABBBA[\ltuple \IDMVBA,   \electionoutput \rtuple]$, which further suggests that at least one honest node inputs    $\ReadyRecord[\electionoutput]=1$ or $\FinishRecord[\electionoutput] =1$ into $\ABBBA[\ltuple \IDMVBA,   \electionoutput \rtuple]$,  based on the biased integrity of $\ABBBA$.    
\item   When one honest node inputs    $\ReadyRecord[\electionoutput]=1$ or $\FinishRecord[\electionoutput] =1$, it is guaranteed that at least $n-2t$ honest nodes have stored correct shares sent from Node~$\electionoutput$ (see Lines~\ref{line:ACDhreadyCond} and \ref{line:ACDhready} of Algorithm~\ref{algm:ACDh}), which implies that every honest node eventually retrieves the same message from $\DataRetrievalh[\ltuple \IDMVBA,  \electionoutput \rtuple]$.            
\end{itemize}
\end{itemize}

\subsection{Analysis of $\OciorMVBAHash$}    \label{sec:AnalysisOciorMVBAHash}

\begin{theorem}  [Agreement]  \label{thm:OciorMVBAHashagreement}
In $\OciorMVBAHash$,  given $n\geq 3t+1$, if any two honest nodes output $\wv'$ and $\wv''$, respectively, then  $\wv'=\wv''$.        
\end{theorem}
\begin{proof}
In $\OciorMVBAHash$, if any two honest nodes output values at Rounds  $\electionround$ and $\electionround'$ (see Line~\ref{line:MVBAHashABAround}  of Algorithm~\ref{algm:OciorMVBAh}), respectively, then $\electionround=\electionround'$. This follows from  the consistency property of the protocols  $\Election[\ltuple \IDMVBA,  \electionround \rtuple]$ and $\ABBA[\ltuple \IDMVBA,   \electionoutput\rtuple]$,  as well as the   consistency property of the   $\DataRetrievalh[\ltuple \IDMVBA,  \electionoutput \rtuple]$ protocol when $\ABBA[\ltuple \IDMVBA,   \electionoutput\rtuple]$ outputs $1$ (see Lemma~\ref{lm:OciorMVBAHashDRconsistency}). 

Moreover,  at  the same round $\electionround$,  if any two honest nodes output $\wv'$ and $\wv''$, respectively, then  $\wv'=\wv''$, due to the consistency property of the protocol   $\DataRetrievalh[\ltuple \IDMVBA,  \electionoutput \rtuple]$ when $\ABBA[\ltuple \IDMVBA,   \electionoutput\rtuple]$ outputs $1$  (see Lemma~\ref{lm:OciorMVBAHashDRconsistency}).  It is worth noting that  $\DataRetrievalh[\ltuple \IDMVBA,  \electionoutput \rtuple]$  is activated only if  $\ABBA[\ltuple \IDMVBA,   \electionoutput\rtuple]$ outputs $1$  (see Line~\ref{line:MVBAHashABAoutput1}  of Algorithm~\ref{algm:OciorMVBAh}). 
\end{proof}

\begin{theorem}  [Termination]  \label{thm:OciorMVBAHashTermination}
In $\OciorMVBAHash$,  given $n\geq 3t+1$, every honest node eventually outputs a value and terminates.        
\end{theorem}
\begin{proof}
In this setting, every honest node eventually returns values and terminates from the protocol $\ACDh[ \IDMVBA ]$, due to the Termination property of this protocol.
Furthermore, by the  Integrity property of  $\ACDh[ \IDMVBA ]$, if an honest node  returns values and terminates from the protocol $\ACDh[ \IDMVBA ]$, then there exists a set $\Ic^{\star}$ such that the following conditions hold: ) $\Ic^{\star}\subseteq [1:n]\setminus \Fc$, where $\Fc$ denotes the set of indexes of all dishonest nodes; 2)  $|\Ic^{\star}| \geq n-2t$; and  3) for any $i\in \Ic^{\star}$,  $\Node_i$ has completed the dispersal $\ACD[\ltuple \IDMVBA, i \rtuple]$.      

Subsequently, every honest node eventually runs  $\electionoutput \gets \Election[\ltuple \IDMVBA,  \electionround \rtuple]$  in Line~\ref{line:MVBAHashElection} of Algorithm~\ref{algm:OciorMVBAh}, at the same round $\electionround$.
If  Node~$\electionoutput$ is honest and  $\electionoutput\in  \Ic^{\star}$, then $\ABBA[\ltuple \IDMVBA,   \electionoutput\rtuple]$ eventually outputs $1$  (due to the biased validity property of  $\ABBBA[\ltuple \IDMVBA,   \electionoutput \rtuple]$) and then during $\DataRetrievalh[\ltuple \IDMVBA,  \electionoutput \rtuple]$ each node will receive at least $t+1$ correct shares	  generated from Node~$\electionoutput$, given  $n\geq 3t+1$. In this case all honest node eventually output the same message  from   $ \DataRetrievalh[\ltuple \IDMVBA,  \electionoutput \rtuple]$, and then terminate. 
If  Node~$\electionoutput$ is dishonest or $\ABBA[\ltuple \IDMVBA,   \electionoutput\rtuple]$ outputs $0$, then all honest nodes proceed to the next round. 
All honest node eventually terminates if, at some round~$\electionround$, $\Election[\ltuple \IDMVBA,  \electionround \rtuple]$  outputs a value $\electionoutput$ such that  Node~$\electionoutput$ is honest and $\electionoutput\in  \Ic^{\star}$. 	
\end{proof}

\begin{theorem}  [External Validity]  \label{thm:OciorMVBAHashExternalvalidity}
In $\OciorMVBAHash$,  if an honest node outputs a value $\MVBAInputMsg$, then $\Predicate(\MVBAInputMsg)=\true$.    
\end{theorem}
\begin{proof}
 In $\OciorMVBAHash$,  if an honest node outputs a value $\MVBAInputMsg$,  it has verified that $\Predicate(\MVBAInputMsg)=\true$ at Line~\ref{line:MVBAHashPredicateCond}  of Algorithm~\ref{algm:OciorMVBAh}.
\end{proof}

 \begin{lemma}     \label{lm:OciorMVBAHashDRconsistency}
In $\OciorMVBAHash$,  if $\ABBA[\ltuple \IDMVBA,   \electionoutput\rtuple]$ outputs $1$, then  every honest node eventually retrieves the same message from $\DataRetrievalh[\ltuple \IDMVBA,  \electionoutput \rtuple]$, for $\electionoutput\in [1:n]$.   
\end{lemma}
\begin{proof}
In $\OciorMVBAHash$,  $\DataRetrievalh[\ltuple \IDMVBA,  \electionoutput \rtuple]$  is activated only if  $\ABAoutput=1$ (see Line~\ref{line:MVBAHashABAoutput1}  of Algorithm~\ref{algm:OciorMVBAh}), where    $\ABAoutput$ is the output of $\ABBA[\ltuple \IDMVBA,  \electionoutput\rtuple]$.   
 The instance of $\ABAoutput=1$ reveals that at least one honest node outputs $\ABBAoutput=1$ from  $\ABBBA[\ltuple \IDMVBA,   \electionoutput \rtuple]$, which further suggests that at least one honest node inputs    $\ReadyRecord[\electionoutput]=1$ or $\FinishRecord[\electionoutput] =1$ into $\ABBBA[\ltuple \IDMVBA,   \electionoutput \rtuple]$,  based on the biased integrity of $\ABBBA$.    
  When one honest node inputs    $\ReadyRecord[\electionoutput]=1$ or $\FinishRecord[\electionoutput] =1$, it is guaranteed that at least $n-2t$ honest nodes have stored correct shares sent from Node~$\electionoutput$ (see Lines~\ref{line:ACDhreadyCond} and \ref{line:ACDhready} of Algorithm~\ref{algm:ACDh}), which implies that every honest node eventually retrieves the same message from $\DataRetrievalh[\ltuple \IDMVBA,  \electionoutput \rtuple]$.  
  \end{proof}



\end{document}